\documentclass[preprint,12pt,authoryear]{elsarticle}

\usepackage{amssymb}
\usepackage{amsmath}

\usepackage[justification=raggedright]{subcaption}

\usepackage{amsthm}
\newtheorem{theorem}{Theorem}

\usepackage{xcolor}

\graphicspath{{images/}}

\definecolor{o}{rgb}{0.8, 0.588, 0.204} 
\definecolor{db}{rgb}{0., 0.445, 0.738} 
\definecolor{b}{rgb}{0.24, 0.45, 0.76} 

\newcommand{\edit}[1]{{\color{black}{#1}}}

\newcommand{\expA}{\emph{PCL}}
\newcommand{\expB}{$\neg s/m$}
\newcommand{\expC}{\emph{local}~$\eta~s/m$}
\newcommand{\expD}{\emph{global}~$\eta~s/m$}

\journal{CAGD special issue on ``Computational Geometric Design''}

\begin{document}

\begin{frontmatter}



\title{Variational Shape Approximation of Point Set Surfaces}


\author[a]{Martin Skrodzki\corref{cor1}}
\author[b]{Eric Zimmermann}
\author[b]{Konrad Polthier}

\cortext[cor1]{Corresponding author: mail@ms-math-computer.science}

\address[a]{ICERM, Brown University, Providence, RI, USA\\ and RIKEN iTHEMS, Wako, Saitama, Japan}
\address[b]{Freie Universit\"at Berlin, Berlin, Germany}

\begin{abstract}	
	\edit{In this work, we present a translation of the complete pipeline for variational shape approximation (VSA) to the setting of point sets.
	First, we describe an explicit example for the theoretically known non-convergence of the currently available VSA approaches.
	The example motivates us to introduce an alternate version of VSA based on a \emph{switch} operation for which we prove convergence.
	Second, we discuss how two operations---\emph{split} and \emph{merge}---can be included in a fully automatic pipeline that is in turn independent of the placement and number of initial seeds.
	Third and finally, we present two approaches how to obtain a simplified mesh from the output of the VSA procedure. This simplification is either based on simple plane intersection or based on a variational optimization problem.
	Several qualitative and quantitative results prove the relevance of our approach.}
\end{abstract}

%


\end{frontmatter}




\section{Introduction}
\label{sec:Introduction}

\noindent Point sets arise naturally in almost all kinds of three-dimensional acquisition processes, like 3D laser-scanning. As early as 1985, they have been recognized as fundamental shape representations in computer graphics by~\cite{levoy1985use}. Ever since, they have found manifold applications e.g.\@ in face recognition, traffic accident analysis, archaeology, and several other fields.

\edit{However, in many applications, large parts of the point set carry redundant information. For example, a flat area of a surface can be sampled sparsely without---compared to an area of high curvature---loosing information. In several applications, it is not even necessary to consider all details carried by the point set. For instance, in architecture---for a first draft---the rough outline of a building suffices and there is no need to send more detailed geometries. In general, when transmitting geometries e.g.\@ to give an overview of a certain portfolio, the general outlines of the geometries suffices and sending only these saves on bandwidths during transmission.
In this sense, algorithms are necessary that reduce a complex geometry to several basic shapes that still retain the most important features of the input. Towards this end,~\cite{cohen2004variational} proposed their \emph{Variational Shape Approximation} (VSA) for meshes.}

\edit{The VSA procedure segments a mesh into a given number of flat proxy regions, see Section~\ref{sec:Method}. Finally, a simplified surface is obtained with only one element for each region, see Section~\ref{sec:Simplification}. A translation of the VSA method to the setting of point sets was done by~\cite{lee2016feature} with the explicit goal of feature curve extraction. While VSA is able to provide an easy to implement simplification of any geometry, it also has several downsides. 
First, it is dependent on the number of proxies which has to be chosen a priori. 
Second, in the previous publications, the quality of the result depends heavily on the manual placement of the starting seeds for the proxies (\cite{cohen2004variational,lee2016feature,yan2006quadraicSurfExVSA}). Towards this end, two manual operations were proposed, which allow for splitting proxy regions of high error and merging neighboring ones with a combined low error~\cite[Sec.~3.5]{cohen2004variational}. 
In the context of meshes, these operations have been automatized~\cite[Sec.~3.1]{yan2006quadraicSurfExVSA}.} 
Third and finally, the previous publications were not able to construct a VSA algorithm with guaranteed convergence. 
This article closes these gaps. 
Our main contributions are:
\begin{itemize}
	\item Providing an example of a growing error during the run of the VSA algorithm which applies to meshes and point sets alike.
	\item Presentation of a modified VSA procedure \edit{including the \emph{switch} operation} and proof of its guaranteed convergence.
	\item \edit{Inclusion of the two operations, \emph{split} and \emph{merge}, as automatic parts in the point set processing pipeline, making the initial choice of a fixed proxy number and the manual selection of seeds unnecessary.}
	\item Extension of variational tangent plane intersection to the setting of point sets and inclusion of the procedure in the VSA pipeline for simplification.
\end{itemize}

This paper is an extension of a chapter in a PhD thesis~\cite[Chapter 5]{skrodzki2019neighborhood}. Some results of the paper have been presented as a poster at the International Geometry Summit 2019 in Vancouver, Canada and have been published in the corresponding poster proceedings, see~\cite{skrodzki2019variational}.



\section{Related Work}
\label{sec:rel_work}

\edit{
\noindent The VSA procedure was introduced by~\cite{cohen2004variational} as a method for concise, faithful approximation of complex three-dimensional meshes. It does so by fitting a set of planar proxies to the input mesh. We will provide a detailed discussion of the procedure in Section~\ref{sec:Method}. As the resulting elements are oriented corresponding to all associated faces of the original mesh, the effects of simplification are less drastic as in the classical approach of~\cite{garland1997surface}. A next step towards even better approximations consisted of the inclusion of more than just planar shapes. In the work of~\cite{wu2005structure}, also e.g.\@ cylinders and spheres are used as proxies to even better approximate the input shape. This was generalized even further by~\cite{yan2006quadraicSurfExVSA} who utilized general quadrics as proxies to be fitted to the input. However, all these methods are implemented in the setting of surface meshes.

A translation of the VSA procedure to the setting of point sets was performed in an article by~\cite{lee2016feature}. 
The authors studied the problem of computing smooth feature curves from CAD type point cloud models.
Their reconstructed curves arise from the intersections of developable strip pairs which approximate the regions along both sides of the features. 
The generation of the developable surfaces is in turn based on VSA.
While the presented results are convincing, it remains unclear whether the approach of fitting developable surfaces works outside of the CAD realm.
Furthermore, the work does not provide details on how to obtain the used linear planar approximations or how to construct a watertight mesh from them.
These aspects motivate the present research.

Our proposed approach incorporates two different areas of point set processing. 
On the one hand, we aim at segmenting the input into several flat---i.e.\@ planar---parts. 
Thus, we will discuss related segmentation approaches in Section~\ref{sec:Segmentation}. 
On the other hand, we want to construct a simplified mesh on the basis of the found flat surfaces patches. 
Therefore, we will present corresponding work on mesh generation and simplification in Section~\ref{sec:MeshingAndSimplifaction}.

\subsection{Segmentation}
\label{sec:Segmentation}

\noindent Segmentation of point clouds is the process of classifying the input into multiple homogeneous regions, where points in the same region will have the same properties.
In real-world applications---like intelligent vehicles, autonomous mapping, and navigation---the problem is challenging because of high redundancy, uneven sampling density, and lack of explicit structure in the input data.
Methods for point set segmentation can roughly be classified as follows: edge-based, region-based (seeded/bottom-up or unseeded/top-down), attribute-based, model-based, graph-based, or machine-learning-based, see \cite{nguyen2013point}, \cite{grilli2017review}.
Following this terminology, the VSA procedure is a seeded, region-based method, which is characterized by starting the segmentation process from seed points and letting regions grow by adding neighbors if they satisfy certain conditions---like normal similarity. 
We refer to the survey of~\cite{nguyen2013point} for a discussion of several corresponding methods. In this work, the authors draw the following conclusion on seeded region-based methods:
\begin{quote}
	\emph{[They] are highly dependent on selected seed points. 
	Inaccurate choosing seed points will affect the segmentation process and can cause under or over segmentation.}
\end{quote}
The survey paper of~\cite{grilli2017review} draws a similar conclusion for the corresponding set of discussed methods.
Hence, in contrast to the procedures covered in the mentioned surveys, we put an emphasis on the independence of both the number and placement of seed points. See Section~\ref{sec:UserControllerLevelOfDetail} for a corresponding discussion.

Point cloud segmentation can be considered either from a semantic or from a geometrical perspective. The former aims at separating a model into its parts: A chair should for instance be segmented into four legs, a seating surface, and a backrest. The geometric approach is to segment the model into different primitives as well as possible. A recent survey paper of~\cite{xie2019review} provides a comprehensive list of methods following both approaches. In the terminology used in this paper, the VSA approach creates a ``plane point cloud segmentation''. While we cannot discuss all works mentioned, we will consider two popular approaches in the following and refer to the survey of~\cite{xie2019review} for a thorough discussion of other related work.

Note that fitting different planar segments to a model can be considered as a natural approach in order to render the model with a reduced set of planar patches. Naturally, those models are captured well that are comprised of mostly planar surface parts. The presence of spherical or cylindrical shapes will cause for larger distortions when approximating only with planar parts. Thus, a next step---after planar fitting---is the usage of other geometric primitives, like spheres, cylinders, cones, or tori. Each of these primitives then require their own fitting. As the VSA approach only fits planes, we briefly discuss different fitting concepts for this primitive. Note at this point that the method of~\cite{wu2005structure} utilizes the exact same procedure for fitting of planes as the original VSA paper by~\cite{cohen2004variational}.

To fit a planar patch, the approach of~\cite{schnabel2007efficient} considers three points~${p_i,p_j,p_\ell\in P}$ from an input point set~$P$ and computes a normal of the plane spanned by these points. This normal is then compared to the respective normals at~$p_i$, $p_j$, and~$p_\ell$. A fitting plane is introduced if all three normal variations stay below a user-given angle. Clearly, the results of this approach heavily depends on the choice of the three points.

In contrast, the approach of~\cite{attene2010hierarchical} places a plane at the weighted barycenter~$b$ of a considered subset~${\{p_i\}\subset P}$ of the input point set~$P$. A weighted covariance matrix is used to determine a normal~$n$ and the fitting error is computed as a weighted least-squares formulation. However, this computation neglects the normal information at the points~$p_i$.

Another choice for the segmentation of point clouds is the algorithm of~\cite{rabbani2006segmentation}. It is popular because of its easily accessible implementation in the widely used Point Cloud Library (PCL) by~\cite{rusu20113d}. This method can be seen as a reduced version of the VSA approach. Regions are also grown from seeds according to normal information. However, the growing process is only executed once and not repeated from a different set of seeds, like in VSA (see Section~\ref{sec:Method}). Thus, the result is even more dependent on the initial seeding than in other, comparable techniques.

For the give reasons, the discussed methods have their respective downsides. Contrasting the presented algorithms, in our translation of the VSA approach, we include the entire normal information of the input point set. Furthermore, we have a setup of the pipeline that ensures independence of the initial seed points, which eliminates the major disadvantage of seed-based region growing methods.
}

\edit{
\subsection{Meshing and Simplification}
\label{sec:MeshingAndSimplifaction}

\noindent As stated in Section~\ref{sec:Introduction}, the ultimate goal of our VSA procedure for point sets is to create a simplified mesh from a set of planar proxy regions. That is, a set of mesh vertices has to be created from the intersection of the proxy planes. Then, these vertices have to be connected to represent face elements for the proxies respectively. While three pairwise non-parallel planes intersect in a unique point, this is not necessarily the case for more than three planes in~$\mathbb{R}^3$. In the context of planar panelization of freeform surfaces,~\cite{zimmer2012variational} confirm this statement, asserting that tangent plane intersection is numerically not stable enough to obtain reliable results. The authors proceed to present a variational approach and a corresponding minimization problem in order to obtain a planar representation of a given mesh structure. This method improves the approach of~\cite{cutler2007constrained} for constrained planar remeshing of architectural geometries, which is itself based on VSA. Therefore, we turn to the work of~\cite{zimmer2012variational} to make the calculation of a simplified mesh from the input point set as robust as possible. See Section~\ref{sec:Simplification} for a discussion of the technical details of the optimization and also for our translation to the setting of point sets.

Aside from VSA, there are other approaches to obtain a simplified mesh from an input point set. For instance, a possibility is to first mesh the input point set and then simplify the created mesh. An overview of methods for meshing of point sets is presented in the survey of~\cite{berger2017survey}. Several methods are available for the subsequent simplification of the mesh. These mostly collapse edges in the mesh to reduce its complexity. By using quadric error metrics, it can be assured that the collapses remove elements that carry the least amount of feature information, see \cite{garland1997surface}. This simple approach can be improved by adjusting the position of the vertex resulting from an edge collapse according to the local curvature information, see \cite{hua2015mesh,yao2015quadratic}. However, as these methods perform their operations in a greedy manner, they do not provide reliable results when performing a drastic number of simplifications. Also, these approaches require a costly meshing operation on the unfiltered point set, which can introduce topological failures, like a surface of wrong genus or flipped triangles.

Another possibility to obtain a simplified mesh from an input point set is to first simplify the point set and to then create a mesh from this. A brief introduction and (error) analysis of different point set simplification algorithms can be found in the work of~\cite{pauly2002efficient}. Important attributes in real-world applications are the performance and quality of the rendering process. This requires a specific focus on features represented by the point sets. By utilizing a bilateral filtering, both Euclidean distances and normal information can be taken into account throughout a simplification process on a point set to best preserve both the large-scale geometry and small-scale features, in~\cite{sosorbaram2010simplification}. While these methods are feature-preserving, they are not robust in the presence of outliers or noise. Also, the construction of the mesh cannot use the full information of the input point set anymore, as the majority of points will have been removed during the simplification step.

Because of the downsides of both the approaches, we aim at taking all points of the input into account when creating planar proxies. From these, we then create a mesh by completely creating new vertices and connections on them without going through a costly meshing operation on the original input, see Section~\ref{sec:Simplification}.
}


\section{The Method}
\label{sec:Method}

\noindent In this section, we will present the Variational Shape Approximation (VSA) as introduced by~\cite{cohen2004variational} and as used by~\cite{yan2006quadraicSurfExVSA} for surfaces and surface meshes. Also, we present a translation of the procedure to the setting of point sets, similar to the work of~\cite{lee2016feature}.

\subsection{The VSA Procedure for Surfaces and Surface Meshes}
\label{sec:TheVSAProcedureForMeshes}

\noindent The VSA procedure of~\cite{cohen2004variational} acts on a surface~${S\subseteq\mathbb{R}^3}$. The goal is to partition~$S$ into~$m$ disjoint regions~${R_i\subseteq S}$, ${\dot{\bigcup}R_i=S}$, where each region is associated a linear proxy~${(C_i,N_i)}$ with a center~${C_i\in\mathbb{R}^3}$ and a unit-length normal~${N_i\in\mathbb{R}^3}$, ${i\in\{1,\ldots,m\}}$. The authors propose two different metrics to find the optimal shape proxies, with the first metric based on the $\mathcal{L}^2$ measure
\begin{align}
\label{equ:L2energy}
	\mathcal{L}^2(R_i,C_i,N_i)=\int_{x\in R_i}\left\|x-\pi_i(x)\right\|_2^2\:dx,
\end{align}
where~$\pi_i(\cdot)$ denotes the orthogonal projection of the argument on the plane with normal~$N_i$ centered at~$C_i$. Thus, the integral~(\ref{equ:L2energy}) measures the squared error between points in the region~$R_i$ and its linear approximation given by~${(C_i,N_i)}$.

A second metric, denoted by~$\mathcal{L}^{2,1}$ is based on the~$\mathcal{L}^2$ measure when evaluated on the normal field. It is given by
\begin{align}
\label{equ:L21energySmooth}
	\mathcal{L}^{2,1}(R_i,N_i)=\int_{x\in R_i}\left\|n(x)-N_i\right\|_2^2\:dx,
\end{align}
where~$n(x)$ denotes the normal of the surface at point~$x\in S$. As~\cite{cohen2004variational} conclude that the~$\mathcal{L}^{2,1}$ metric is more effective, we will reduce the following discussion to this formulation.

In the discrete setting, the surface~$S$ is given by a finite set of~${T\in\mathbb{N}}$ (triangular) elements~$t_j$, ${j\in[T]}$ and the centers~$C_i$ are found by randomly choosing a triangle~$t_j$ as center~$C_i$. Therefore, the second smooth formulation~(\ref{equ:L21energySmooth}) can be discretized to
\begin{align}
\label{equ:L21energy}
	\mathcal{L}^{2,1}(R_i,N_i)=\sum_{t_j\in R_i}\left\|n(t_j)-N_i\right\|_2^2\cdot|t_j|,
\end{align}
with~$n(t_j)$ the normal and~$|t_j|$ the area of the element~$t_j$ respectively.

The actual minimization of expression~(\ref{equ:L21energy}) with respect to the segmentation of~$S$ into regions~$R_i$ and with respect to the proxies~${(C_i,N_i)}$ is then performed iteratively. For this, a variation of Lloyd's fixed point iteration given by~\cite{lloyd1982least} is used. The first step is to pick a user-given number~$m$ of center elements~$C_1,\ldots,C_m$ randomly from the set of triangles~${\{t_j\mid j\in[T]\}}$. The normals~$N_i$ are set to the normals of corresponding center triangles~$C_i$ and the regions are initialized as~${R_i=\{t_i\}}$. The neighbors of the chosen center triangles are collected in a priority queue~$\mathcal{Q}$ sorted increasingly with growing $\mathcal{L}^{2,1}$-distance between neighboring triangle and center triangle:~${\left\|n(t_j)-N_i\right\|_2^2}$. Then, the following three steps are performed iteratively until convergence:
\begin{enumerate}
	\item \emph{Flood}: As long as the queue~$\mathcal{Q}$ is not empty, pop the first element~$t_j$ from~$\mathcal{Q}$. Ignore it, if it has already been assigned to a region. If it is not assigned yet, assign it to the region~$R_i$ that pushed it into the queue and push all neighboring elements of~$t_j$ into~$\mathcal{Q}$, noting that they have been pushed by~$R_i$. Without loss of generality, we assume~$S$ to be connected. If that is not the case, the algorithm can simply be run on each connected component of~$S$. For a connected surface~$S$, after the queue~$\mathcal{Q}$ has been emptied, all elements~${\{t_j\mid j\in[T]\}}$ have been assigned to some region respectively.
	\item \emph{Proxy Update}: The proxy normals~$N_i$ are updated according to
	\begin{align*}
		N_i=\frac{\sum_{t_j\in R_i}|t_j|n(t_j)}{\left\|\sum_{t_j\in R_i}|t_j|n(t_j)\right\|_2},
	\end{align*}
	where it is ensured that the updated~$N_i$ are unit-length normals. Note that as the surface will be segmented into a large enough number of locally flat patches, the denominator of this expression will never be zero in practice.
	\item \emph{Seed}: For each region~$R_i$, find some element~${t'\in R_i}$ such that
	\begin{align*}
		\left\|n(t')-N_i\right\|_2^2\leq\left\|n(t_j)-N_i\right\|_2^2
	\end{align*}
	for all~${t_j\in R_i}$. This ensures that the flooding in the next iteration is started from regions that best reflect the current proxy normals.
\end{enumerate}
Finally, the iteration is stopped, when no region changes from one step to the next. From the converged regions~$R_i$ and assigned proxies~${(C_i,N_i)}$, a simplified mesh with corresponding~$m$ surface elements is constructed. Respective results are shown in~\cite{cohen2004variational}.


\subsection{The VSA Procedure on Point Sets}
\label{sec:TheVSAProcedureOnPointSets}

\noindent We will now proceed to present a translation of the VSA procedure to the setting of point sets. A corresponding reformulation can be found in~\cite{lee2016feature}, while we include weights to obtain a more general setup. Compared to the VSA on meshes, several details have to be adjusted for the method to work on point sets. At first, consider the partition problem as stated in Section~\ref{sec:TheVSAProcedureForMeshes}. In the context of point sets, not elements, but the points themselves have to be assigned to the proxies. That is, the given point set ${P=\{p_1,\ldots,p_\aleph\}}$ will be partitioned into disjoint subsets ${\dot{\bigcup}_{i=1}^{m} P_i=P}$,~${m\in\mathbb{N}}$. Therefore, in the following expressions, the centers~$C_i$ denote points from the point set~$P$, while the normals~$N_i$ at the respective center point are those obtained from a normal field imposed on the point set. The normals of the points~${p_j\in P}$ will be denoted by~$n_j$ respectively.

Consider the energy as defined in Equation~(\ref{equ:L21energy}). For proxies obtained from point sets, the area term~$|t_j|$ cannot be used. Thus, we replace it by a weighting term~${\omega_j\in\mathbb{R}_{\geq0}}$ which is to approximate the area represented by the point~${p_j\in P}$. We obtain the following energy of a single proxy and the resulting energy formulation on the set of all proxies
\begin{align}
\label{equ:L21energyPointsOneProxy}
	\mathcal{L}^{2,1}(P_i,N_i) &= \sum_{p_j\in P_i}\omega_{j}\left\|n_j-N_i\right\|_2^2,\\
\label{equ:L21energyPoints}
	E(\{(P_i,N_i)\mid i=1,\ldots,m\}) &= \sum_{i=1}^m \mathcal{L}^{2,1}(P_i,N_i).
\end{align}
\edit{An approximation of the area term can be obtained by 
\begin{align}
\label{equ:DiscreteAreaWeights}
	\omega_j=\sum_{\ell\in\mathcal{N}(j)}\left\|p_{\ell}-p_j\right\|_2^2,
\end{align}
where~${\mathcal{N}(j)\subset P}$ denotes the neighborhood of $p_j$ in $P$. Including weights reflecting the area are of interest because of varying densities. Therefore, another possible weighting scheme could be the incorporation of directional density measures proposed in~\cite{skrodzki2018directional}. This could be coupled in a bilateral manner together with the Euclidean distances mentioned before. In contrast, weight determination via normal deviations should not be used, as the energy is defined upon these, i.e. weighting these terms in the same fashion seems counterproductive.}

The initial seeding as outlined above can still be done in the point cloud setting, but instead of triangles, now,~$m\in\mathbb{N}$ points~${p_j\in P}$ are chosen for the initial position of the center points~$C_i$. Also, those points from~$P$ are pushed to the priority queue~$\mathcal{Q}$ that are neighbors, but not identical, to the chosen center points~$C_i$. For this neighborhood relation, any neighborhood concept such as combinatorial $k$-nearest neighborhoods or geometric neighborhoods of radius~$r$ can be used. Denote the neighborhood of~$p_i$ by~${\mathcal{N}(i)\subset P}$. Again, the points in~$\mathcal{Q}$ are sorted increasingly with~$\mathcal{L}^{2}$ distance between their own normal and the normal of the proxy that pushed them into the queue:~${\left\|n_j-N_i\right\|_2^2}$. The following three iteratively applied steps remain almost unchanged:

\begin{enumerate}
	\item \emph{Flood}: As long as the queue~$\mathcal{Q}$ is not empty, pop the first element~$p$ from~$\mathcal{Q}$. Ignore it, if it has already been assigned to a subset~$P_i$. If it is not assigned yet, assign it to the subset~$P_i$ that pushed it into the list and push all neighboring points~${p_j\in\mathcal{N}(i)}$ into~$\mathcal{Q}$, noting that they have been pushed by~$P_i$. As we assume~$S$ to be  connected via the imposed neighborhood relation (see above), after the queue~$\mathcal{Q}$ has been emptied, all elements of~$P$ have been assigned to some subset~$P_i$.
	\item \emph{Proxy Update}: The proxy normals~$N_i$ are updated according to
	\begin{align*}
		N_i=\frac{\sum_{p_j\in P_i}\omega_{j}n_j}{\left\|\sum_{p_j\in P_i}\omega_jn_j\right\|_2},
	\end{align*}
	where we once again obtain unit-length normals and will not encounter a denominator equal to zero (see above).
	\item \emph{Seed}: For all subsets~$P_i$, find some~${p_\ell\in P_i}$, ${\ell\in[\aleph]}$, such that
	\begin{align*}
		\left\|n_\ell-N_i\right\|_2^2\leq\left\|n_j-N_i\right\|_2^2
	\end{align*}
	for all ${p_j\in P_i}$. Again, this ensures that the next flooding step starts from regions that best reflect the current proxy normals.
\end{enumerate}
Finally, once the subsets~${P_i}$ do not change anymore over two iterations, the process is stopped. From the converged subsets~$P_i$ and assigned proxies~${(C_i,N_i)}$, a simplified mesh with corresponding~$m$ surface elements is constructed. Respective results are shown in~\cite{lee2016feature}, while our corresponding approach will be discussed in Section~\ref{sec:Simplification}.


\section{Improved VSA Pipeline}
\label{sec:ImprovedVSAPipeline}

\noindent Having described the VSA procedure for both meshes and point sets in the previous chapter, we now turn to our contributions for this pipeline. First, we will establish by an example that convergence of neither the meshed nor the point set version is guaranteed. \edit{Following up on this, we propose an alternative formulation of VSA with guaranteed convergence.} Furthermore, we turn to a different issue of the VSA procedure. Namely, it is highly depended on both the number of initial seeds and their placement at the beginning of the procedure. 
\edit{We circumvent this dependency by including two more operations in the point set pipeline that have already been used manually~\cite{cohen2004variational} and automatically~\cite{yan2006quadraicSurfExVSA} in the context of meshes.}

\subsection{Example for Failure of Convergence of the VSA Procedure}
\label{sec:ExampleConvergenceFailure}

\noindent Concerning the convergence of their algorithm,~\cite{cohen2004variational} state:
\begin{quote}
	\emph{[$\ldots$] Lloyd's algorithm always converges in a finite number of steps, since each step reduces the energy~$E$: the partitioning stage minimizes~$E$ for a fixed set of centers~$c_i$, while the fitting stage minimizes~$E$ for a fixed partition.}
\end{quote}	
While this statement holds for the original algorithm of Lloyd as presented in~\cite{lloyd1982least}, it does not hold for neither the VSA procedure on meshes as presented in~\cite{cohen2004variational,yan2006quadraicSurfExVSA} nor for the translation to point sets as given by~\cite{lee2016feature}. This is already recognized in the paragraph \emph{Convergence} of Section~3.5 in~\cite{cohen2004variational}. We will demonstrate this with the following concrete example, which is to the best of our knowledge the first explicit example presented. 

\begin{figure}
	\captionsetup[subfigure]{justification=centering}
	\begin{subfigure}[b]{0.5\textwidth}
		\def\svgwidth{1.\textwidth}
		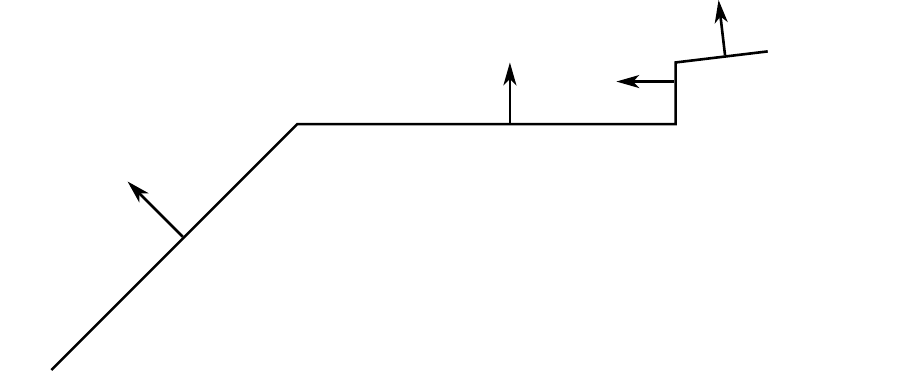
		\caption{Setup for growing error functional.}
		\label{fig:ConvergenceExampleSetup}
	\end{subfigure}
	\begin{subfigure}[b]{0.24\textwidth}
		\def\svgwidth{1.\textwidth}
		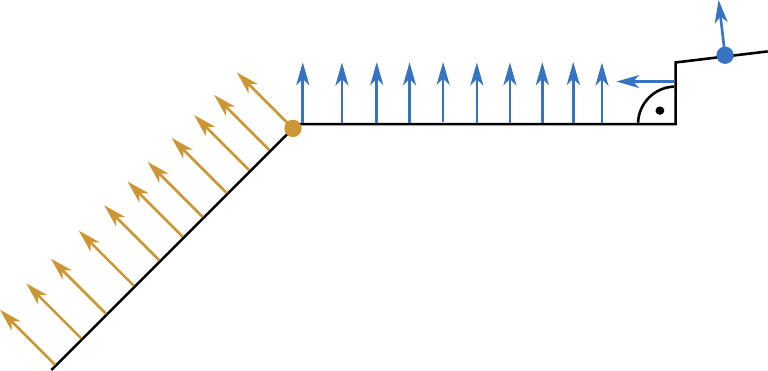
		\caption{Segmentation after first flood.}
		\label{fig:ConvergenceExample1Flood}
	\end{subfigure}
	\hfill
	\begin{subfigure}[b]{0.24\textwidth}
		\def\svgwidth{1.\textwidth}
		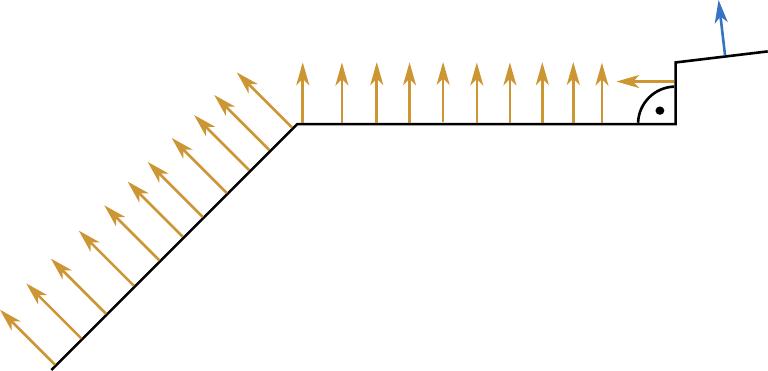
		\caption{Segmentation after second flood.}
		\label{fig:ConvergenceExample2Flood}
	\end{subfigure}
	\caption{Example for a growth in the error measure after a flood and proxy update.}
\end{figure}

Consider the two-dimensional setup shown in Figure~\ref{fig:ConvergenceExampleSetup}. It is given by~$n$ points connected on a line with normal~\edit{$ {\frac{1}{\sqrt{2}}{-1\choose 1}} $} next to a line of~$n$ points with normal~${0\choose 1}$. At the right end of the second line, there is a single point with normal~${-1\choose 0}$ and another single point with normal~$N$ given by the equation
\begin{align*}
	N = \edit{\frac{1}{\left\|n\cdot{0 \choose 1} + {-1 \choose 0} + N\right\|_2}}\cdot \left(n\cdot{0 \choose 1} + {-1 \choose 0} + N\right),
\end{align*}
\edit{which solves to~$ { N = \frac{1}{\sqrt{n^2+1}}{-1 \choose n}} $}.
Now, two proxies will act on this example, with their initial seeds shown in yellow and blue in Figure~\ref{fig:ConvergenceExampleSetup}. They each start on one of the two lines of~$n$ points respectively. The result after a flood is shown in Figure~\ref{fig:ConvergenceExample1Flood}, where each line is completely covered by the proxy starting on it and the two single points are associated to the proxy with normal~${0 \choose 1}$. After updating the proxy normals, the yellow proxy has normal~\edit{$ {\frac{1}{\sqrt{2}}{-1 \choose 1}} $} while the blue proxy has normal~$N$ given by the equation above. Thus, the yellow proxy starts from an arbitrary point on its line while the blue proxy starts from the rightmost point. The error after this first flood and proxy update is given by
\begin{align*}
	E_1 = n\cdot\left\|{0 \choose 1} - N\right\|_2^2 + \left\|{-1 \choose 0} - N\right\|_2^2 = -2(\sqrt{n^2+1}-n-1),
\end{align*}
\edit{where only the blue proxy contributes to the error, because the normals corresponding to the yellow proxy coincide with their proxy normal and cancel out in energy $ E_1 $.}
Starting from the new seed points, a second flood results in the situation shown in Figure~\ref{fig:ConvergenceExample2Flood}. Here, almost all points except for the rightmost one are associated to the yellow proxy with former normal~\edit{$ {\frac{1}{\sqrt{2}}{-1 \choose 1}} $}. Its new normal after a proxy update is
\begin{align*}
	N' = \edit{\frac{1}{\left\| \frac{n}{\sqrt{2}}\cdot {-1 \choose 1} + n\cdot {0 \choose 1} + {-1 \choose 0}\right\|}}\cdot\left( \frac{n}{\edit{\sqrt{2}}}\cdot {-1 \choose 1} + n\cdot {0 \choose 1} + {-1 \choose 0}\right),
\end{align*}
which amounts to an error after the second flood and proxy update given by
\begin{align*}
	E_2 = n\cdot\left\|\edit{\frac{1}{\sqrt{2}}{-1 \choose 1}} - N'\right\|^2_2 + n\cdot\left\|{0 \choose 1} - N'\right\|^2_2 + \left\|{-1 \choose 0} - N'\right\|^2_2.
\end{align*}
\edit{Note that the error term for the blue proxy cancels out, as the one representative corresponds to the normal of the proxy it belongs to.} Choosing~$n=100$ points on each of the two lines, we obtain~$ \edit{{E_1\approx1.9900}} $, but~$ \edit{{E_2\approx31.6782}} $. Furthermore, the corresponding error value after the flood is also growing. Thus, convergence cannot be proven by an always shrinking error functional. 

\edit{Note that this example is described as a curve in 2D, where neighborhood selection is generally more involved than for surfaces in 3D. However, the example can easily be extended to a surface in 3D space, see Figure~\ref{fig:ConvergenceExample3D}. There, we also close the loop and thereby cause the original VSA algorithm to run infinitely long. For the given example, the crucial step as depicted in Figure~\ref{fig:ConvergenceExample2Flood} can be resolved via a manual~(\cite{cohen2004variational}) or automatic~(\cite{yan2006quadraicSurfExVSA}) split of the large proxy. Thus, this example only applies to the VSA procedure as described above.}

\begin{figure}
	\centering
	\includegraphics[width=1.0\textwidth]{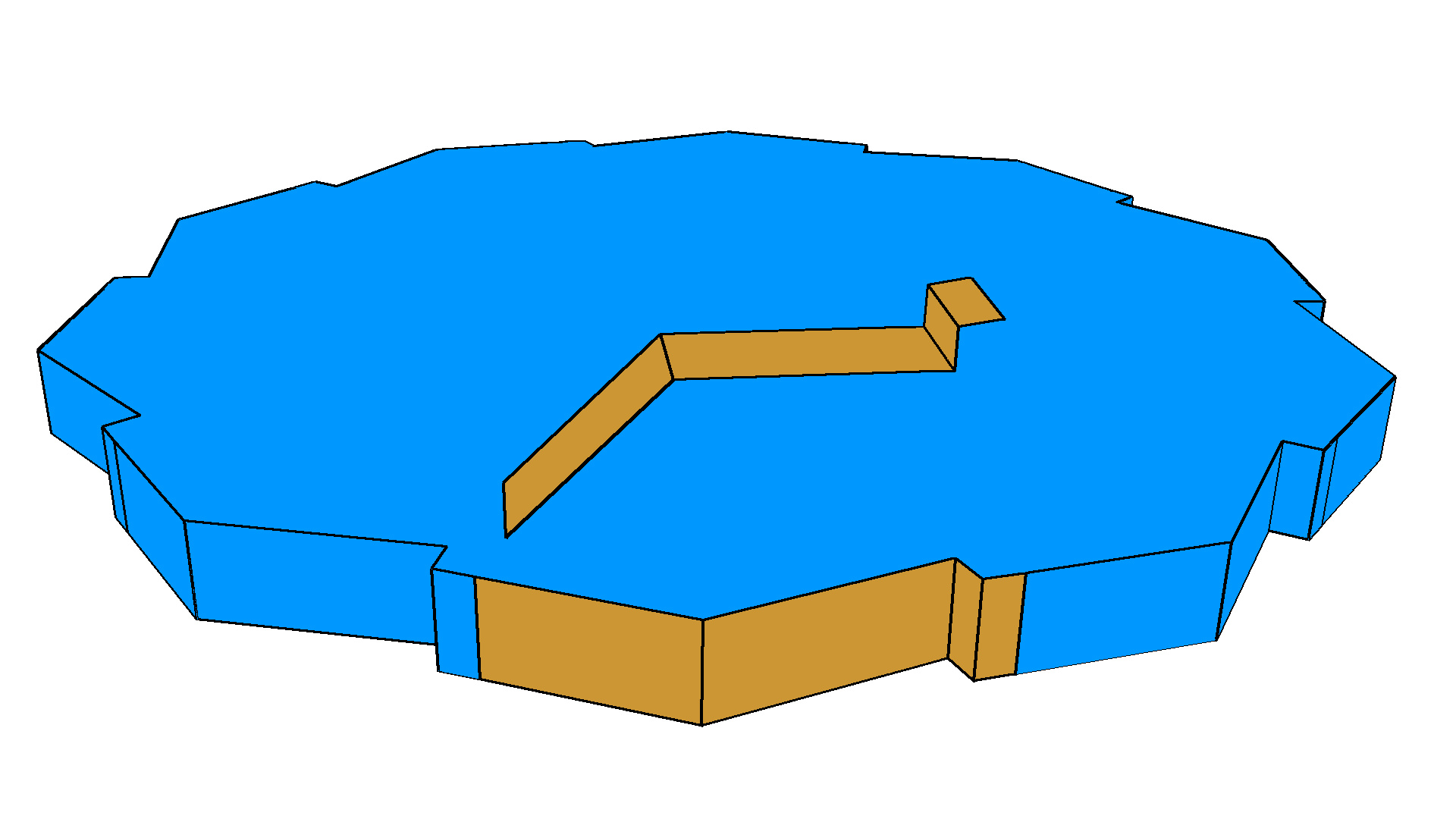}
	\caption{A regular 10-gon, built from the shape shown floating on top, which is a three-dimensional extension of the setup shown in Figure~\ref{fig:ConvergenceExampleSetup}.}
	\label{fig:ConvergenceExample3D}
\end{figure}


\subsection{VSA with guaranteed Convergence}
\label{sec:VSAWithGuaranteedConvergence}

\noindent \edit{The example presented above highlights the main deficiency of the VSA procedure as used in~\cite{cohen2004variational,yan2006quadraicSurfExVSA,lee2016feature}. 
Namely, if an outlier causes a proxy normal to be distorted, the new proxy seed can end up to be a border point that does not actually reflect the normal behavior of the majority of points in the proxy.
In other words, the change of seeds before flooding is a problematic step.
Thus, in the following, we aim at altering the VSA procedure in a way such that no new seeds need to be found, but proxies can still move and alter.
In particular, proxies should be able to take over the original seed points of other proxies if necessary.
These changes should finally lead to an alternative VSA procedure with guaranteed convergence.
In order to achieve this goal, we propose to alter the steps of the algorithm as follows.}

First, we perform an initial seeding and one flood step and proxy update as explained in Sections~\ref{sec:TheVSAProcedureForMeshes} and~\ref{sec:TheVSAProcedureOnPointSets} above. Instead of the seeding step in the following iterations, we perform a different procedure:
\begin{enumerate}
	\item[4.] \emph{Switch}: Consider the neighborhoods~${\mathcal{N}(i)\subset P}$ for all points~${p_i\in P}$. Assume that~$p_i$ is assigned to subset~$P_\ell$. If any point~$p_j\in\mathcal{N}(i)$ is assigned to another subset~$P_h$, compute the change of the error measure~(\ref{equ:L21energyPoints}) resulting from reassigning~$p_i$ from~$P_\ell$ to~$P_h$. Compare it to the current best known reassignment. After iterating through all points~${p\in P}$, reassign the point such that the error measure is reduced maximally.
\end{enumerate}
This new switch step replaces the \emph{seed} step and the \emph{flood} step described in Sections~\ref{sec:TheVSAProcedureForMeshes} and~\ref{sec:TheVSAProcedureOnPointSets} above. That is, it is only iterated together with the \emph{proxy update}. The iteration is continued until no further switch operations can be performed. For this alternate procedure, we can prove the following statement.

\begin{theorem}[Error reduction by switch and proxy update]
	\label{the:ConvergenceOfVSA}
	Given a point set ${P=\{p_1,\ldots,p_\aleph\}}$ with a neighborhood structure, such that the neighborhood graph on~$P$ is connected and normals~${n_1,\ldots,n_\aleph}$ on~$P$. Then, each proxy update step and each switch step as defined above leads to proxies~${(P_i,C_i,N_i)}$ with a smaller error measure in Equation~(\ref{equ:L21energyPoints}).
\end{theorem}
\begin{proof}
	Concerning the proxy update step, consider
	\begin{align*}
	\nabla E(\{(P_i,N_i)\mid i\in[m]\}) &= \nabla\sum_{i=1}^m\mathcal{L}^{2,1}(P_i,N_i)\\
	&= \sum_{i=1}^m\sum_{p_j\in P_i}\nabla\omega_j\left\|n_j-N_i\right\|_2^2\\
	&= \sum_{i=1}^{m}\sum_{p_j\in P_i}2\omega_j(n_j-N_i).
	\end{align*}
	Setting $N_i=\frac{\sum_{p_\ell\in P_i}\omega_\ell n_\ell}{\sum_{p_\ell\in P_i}\omega_\ell}$, we obtain
	\begin{align*}
	\sum_{p_j\in P_i}2\omega_j(n_j-N_i) 
	&= \sum_{p_j\in P_i}2\omega_jn_j-\sum_{p_j\in P_i}2\omega_j\left(\frac{\sum_{p_\ell\in P_i}\omega_\ell n_\ell}{\sum_{p_\ell\in P_i}\omega_\ell}\right)\\
	&= \sum_{p_j\in P_i}2\omega_jn_j - \left(\frac{\sum_{p_\ell\in P_i}2\omega_\ell n_\ell}{\sum_{p_\ell\in P_i}\omega_\ell}\right)\cdot\sum_{p_j\in P_i}\omega_j\\
	&= \sum_{p_j\in P_i}2\omega_jn_j - \sum_{p_\ell\in P_i}2\omega_\ell n_\ell =0.
	\end{align*}
	Thus, at the chosen updated proxy normal, the energy reaches a (local) minimum. As the energy is convex as sum of norms, which are convex, the found minimum is indeed its global minimum for the current choice of segmentation.
	
	Concerning the switch step, only those points are reassigned which reduce the value of error measure~(\ref{equ:L21energyPoints}). Thus, trivially, after a switch operation the error is smaller.
\end{proof}
\noindent Finally, we note that there are only finitely many ways to partition the~$\aleph$ points of the point set~$P$ into~$m$ subsets. This fact, together with Theorem~\ref{the:ConvergenceOfVSA} proves the convergence of our modified VSA procedure.

\edit{Consider the application of this alternative VSA version to the setup in Figure~\ref{fig:ConvergenceExampleSetup}. After a first flood, which would still result in the proxies shown in Figure~\ref{fig:ConvergenceExample1Flood}, the only possible switch could be performed at the border between the blue and the yellow region. However, a switch would already lead to an increase of the energy functional. Thus, the proxies remain as they are after the first flood and the example converges immediately.}

\edit{By replacing \emph{seed} and \emph{flood} with the \emph{switch} operation, we can ensure convergence of the algorithm.
While this result is theoretically pleasing, it is not necessarily of practical value.
Finding an ideal pair of points for a switch operation requires to iterate at least over all points on the border of proxy regions.
Depending on the number of proxies and on the shape of the geometry, one such switch can reach the same time complexity as a flood operation while only altering a single point's proxy assignment. Thus, in practice, utilizing the \emph{switch} operation causes a significantly longer runtime as trade-off to the guaranteed convergence.}

\edit{Furthermore, iterated application of the switch operation can tear a proxy apart, see Figure~\ref{fig:switch}.
A converged state of the algorithm might therefore include proxy regions that are not connected.
In order to have a sensible result, a final step has to be included that re-interprets connected regions as proxies and that might increase the number of proxies doing so.
However, a disconnectedness only arises if another proxy better reflects the local shape.
Thus, a corresponding higher number of connected proxy regions is desirable in order to faithfully approximate the input geometry.}

\begin{figure}
	\centering
	\captionsetup[subfigure]{justification=centering}
	\begin{subfigure}[t]{0.49\textwidth}
		\includegraphics[width=1.\textwidth]{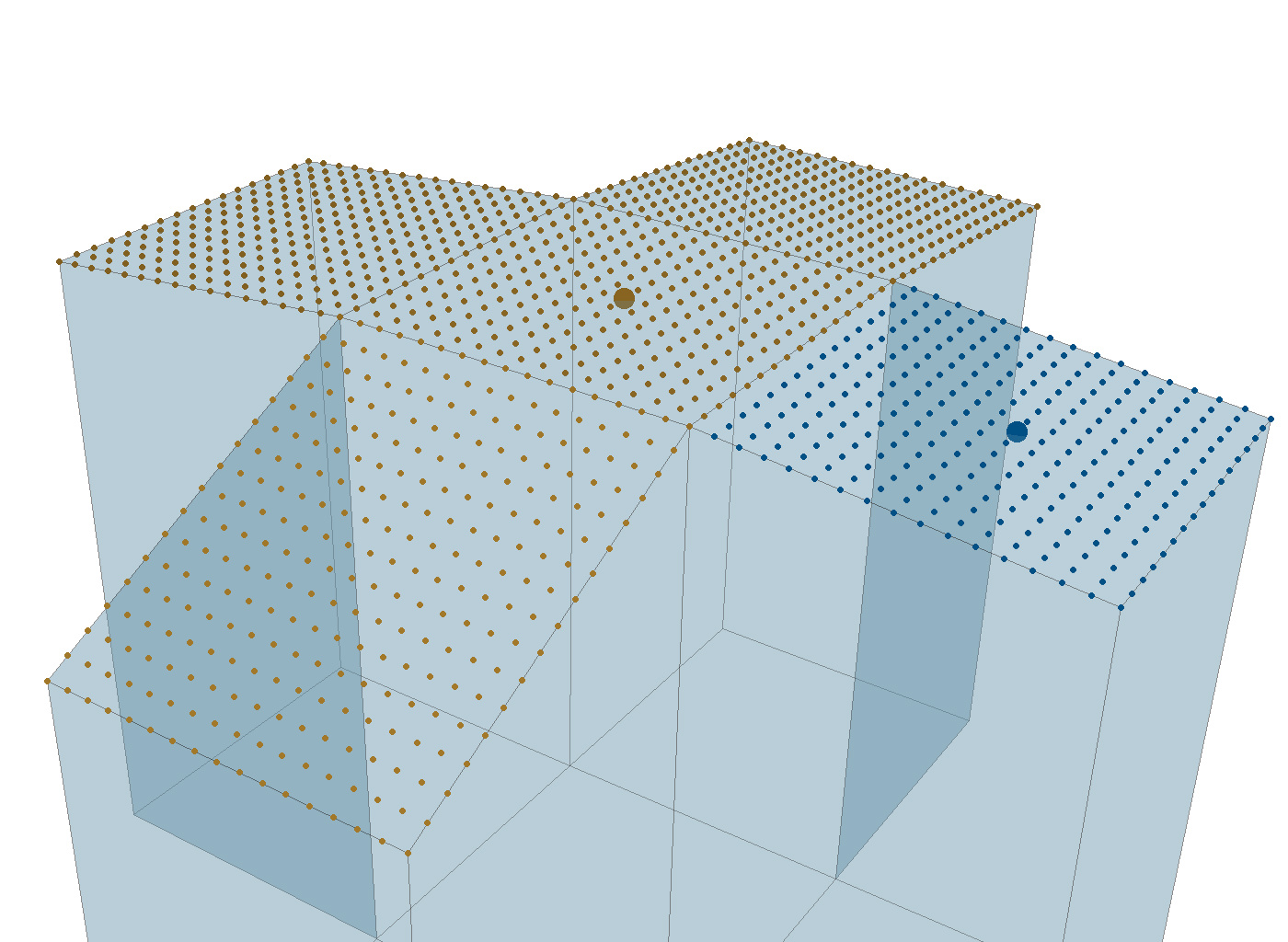}
		\caption{A geometry after initial selection of seeds as indicated and flooding.}
		\label{fig:switchFlood}
	\end{subfigure}
	\hfill
	\begin{subfigure}[t]{0.49\textwidth}
		\includegraphics[width=1.\textwidth]{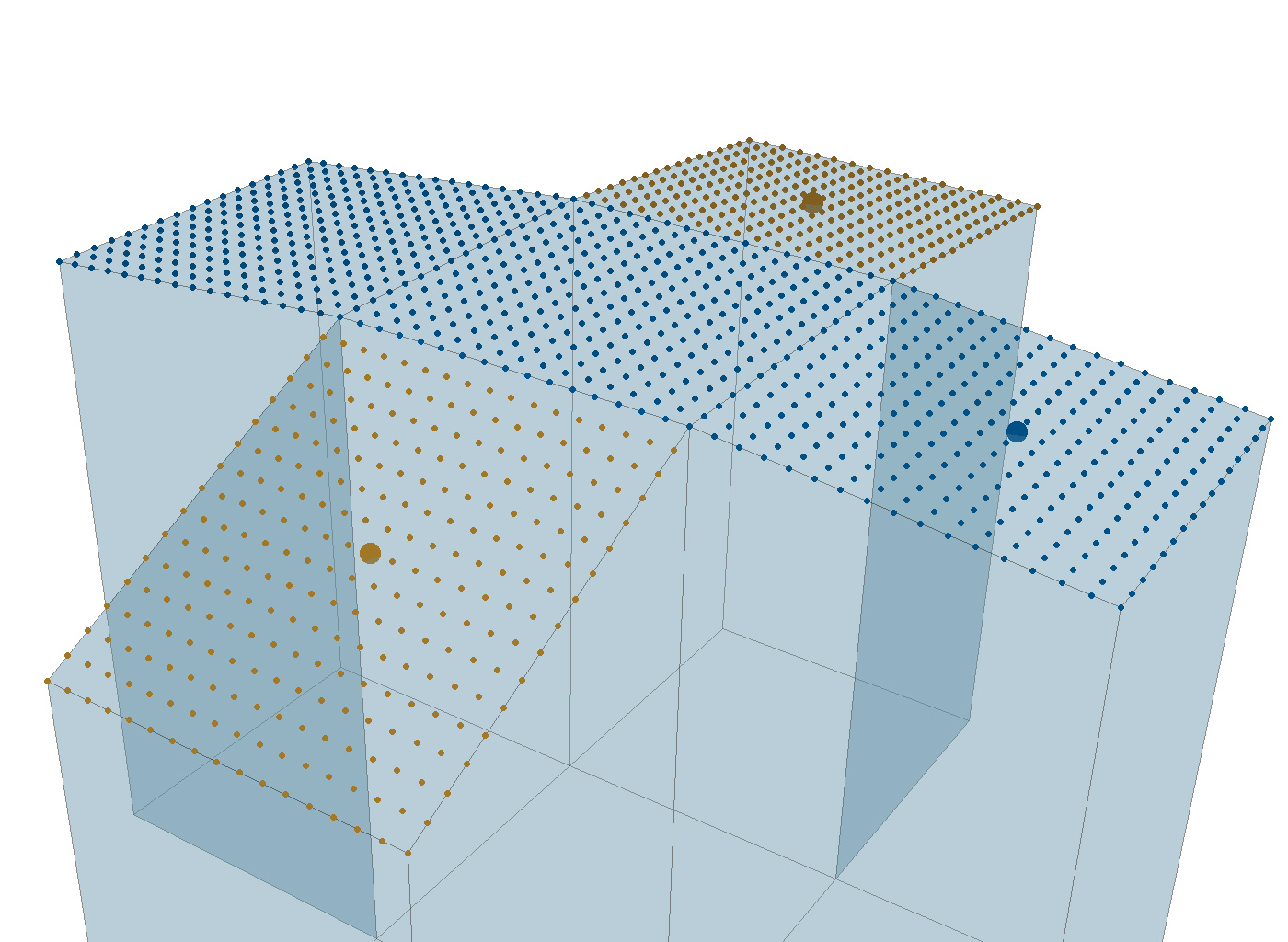}
		\caption{The same geometry after several switch operations. The blue proxy split the other proxy in two components.}
		\label{fig:switchSwitches}
	\end{subfigure}
	\caption{A proxy being torn apart by another proxy under the repeated application of the switch operation.}
	\label{fig:switch}
\end{figure}

\edit{The presented \emph{switch} operation provides one possible way to obtain guaranteed convergence. 
It remains as open question whether another operation or alteration of the VSA procedure provides the same result while coming with a lower runtime.}


\subsection{User controlled Level of Detail}
\label{sec:UserControllerLevelOfDetail}

\noindent \edit{The requirements of proper seed placement and prescribed seed number naturally demand for a variable proxy-treatment in terms of \emph{splits} and \emph{merges}. Both concepts were introduced in~\cite{cohen2004variational} as means for manual adjustments by the user. For meshes, these two operations are incorporated into the pipeline described in~\cite{yan2006quadraicSurfExVSA}. In the following, we propose a translation to point sets. With both operations,} we aim at adaptability of the constructed flat pieces towards user input. 
That is, the user should be able to control the level of detail obtained from the flat regions. 
\edit{However, in contrast to~\cite{cohen2004variational}, this control should be realized via a single input parameter instead of a time-consuming manual interaction with the modeling process.}
For this, we use a user-given parameter~${\eta\in\mathbb{R}_{\geq0}}$ which controls the maximum deviation of a subset~$P_i$ from its flat approximation. It can be thought of as controlling the maximum bending of a segment. This parameter is used in the following two steps:

\begin{enumerate}
	\item[(a)] \emph{Split}: Given a subset~${P_i\subset P}$ such that~${\mathcal{L}^{2,1}(P_i,N_i)>\eta}$. We use weighted principal component analysis by~\cite{harris2011geographically} to compute the most spread direction of~$P_i$. The set~$P_i$ is then split at the center of this direction into two new sets~${P_i=P_i^1\dot{\cup}P_i^2}$. The new normals are chosen as~${N_i^1=\sum_{p_j\in P_i^1}\frac{\omega_j n_j}{\sum_{p_j\in P_i^1}\omega_j}}$ and~$N_i^2$ respectively. The new centers~$C_i^1$ and~$C_i^2$ are then placed at those points of~$P_i^1$ and~$P_i^2$ that have least varying normals from~$N_i^1$ and~$N_i^2$ respectively.\\
	Note that the reasoning of Theorem~\ref{the:ConvergenceOfVSA} holds for this case, too. Thus, the procedure outlined above, with an additional split step does continue to converge.
	\item[(b)] \emph{Merge}: Consider a pair~$P_i$,~$P_j$ of neighboring subset with their respective normals~$N_i$,~$N_j$. If the subset~${P'=P_i\cup P_j}$ with normal
	\begin{align*}
		N'= \left\|\frac{|P_i|\cdot N_i+|P_j|N_j}{|P_i|+|P_j|}\right\|^{-1} \frac{|P_i|\cdot N_i+|P_j|N_j}{|P_i|+|P_j|}
	\end{align*}
	achieves an Energy~(\ref{equ:L21energyPoints}) strictly less than~$\eta$, the two subsets are replaced by their union~$P'$, with normal~$N'$ and a chosen center~${C'\in P'}$ with its normal least deviating from~$N'$.\\
	Note that we could allow only those pairs of neighboring regions to merge such that
	\begin{align*}
		\mathcal{L}^{2,1}(P_i,N_i)+\mathcal{L}^{2,1}(P_j,N_j)\geq\mathcal{L}^{2,1}(P',N').
	\end{align*} 
	Then, the energy would not increase and termination of the algorithm would be guaranteed by Theorem~\ref{the:ConvergenceOfVSA}. However, this would result in neighboring regions not observing the user-given~$\eta$ threshold. Therefore, we accept an increase of the global energy in favor of a better region layout\footnote{Note that the equation for $N'$ in this description of the \textit{merge} procedure deviates from the equation given in the published version of the article in \textit{Computer Aided Geometric Design 2020, Vol. 80}. The formulation given here is more general and works in particular if~$P_i$ and~$P_j$ are of different sizes. Furthermore, in the published version, the inequality on~$\mathcal{L}^{2,1}$ was given in the wrong direction, which is corrected here.}.
\end{enumerate}
Both operations alter the number~$m$ of proxies. Thereby, a significant disadvantage of the algorithm of~\cite{lloyd1982least} is eliminated as the user does not have to choose~$m$ a priori. It is replaced by the user's choice of~$\eta$, providing a semantic guarantee on the regions being built by the algorithm. The user can prescribe a value of~$\eta$ based on the curvature and number of points within a proxy. \edit{See~\ref{app:InterpretationOfEta} for a more detailed interpretation of~$\eta$.}

The possible presence of noise in the point set~$P$ gives yet another reason to refute Energy~(\ref{equ:L2energy}). For points distributed around the~$xy$-plane, with normals~$(0,0,1)^T$ and just slight deviation from the plane, this energy would create larger values for a growing number of points, while the Energy~(\ref{equ:L21energyPoints}) does not suffer from this. Hence, with the chosen energy, noise on the point positions is handled more robustly.

In the merge process outlined above, we asked for two neighboring regions. However, we have not defined any relation on the regions yet. In the meshed case discussed in Section~\ref{sec:TheVSAProcedureForMeshes}, two regions are neighbors if and only if they share an edge in the mesh. In the context of point sets, we cannot rely on this, thus we have used the following approach. \edit{For every proxy and each of its points, we query~$ k $ of the point's nearest neighbors and use the distance to the farthest of them for a geometric neighborhood determination. From all the neighbors gathered that way, we ask for their proxy assignment. If the current center point is assigned to a different proxy than its neighbors, we consider the two proxies to be neighbors.}
This finishes the whole pipeline, including the additional two steps \emph{merge} and \emph{split}. See Figure~\ref{fig:pipeline} for an illustration of the complete pipeline.

\begin{figure}
	\centering
	\includegraphics[width=1.\textwidth]{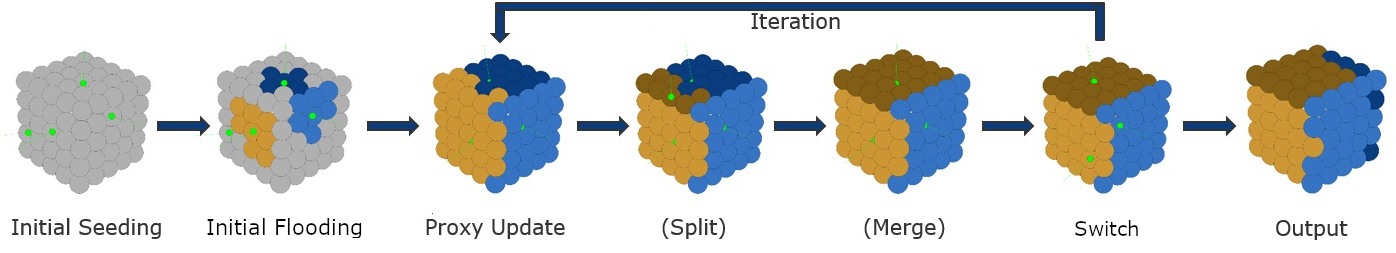}
	\caption{The whole pipeline contains an initial random seed selection and an initial flooding. From there, a proxy update, one or more optional splits and/or merges, and a switch are iterated until no further switches can be applied. Afterwards, we deduce a simplified model according to the proxies, which can also be considered as a simplified surface reconstruction from the initial point set.}
	\label{fig:pipeline}
\end{figure}


\section{Simplification}
\label{sec:Simplification}

\noindent \edit{We will now investigate the creation of a simplified mesh based on the segmentation generated before.} Both works of~\cite{cohen2004variational} and~\cite{lee2016feature} present simplified meshed geometries with~$m$ faces, one representing each proxy. The work of~\cite{yan2006quadraicSurfExVSA} also presents simplified meshes, utilizing their proxy quadrics. However, the approaches of~\cite{cohen2004variational} and~\cite{yan2006quadraicSurfExVSA} are not suitable for our context as they work on meshes. The authors of~\cite{lee2016feature} do not elaborate on the computation of their meshes. They only state that
\begin{quote}
	\emph{($\ldots$) a polygonal mesh is easily generated by computing intersections of proxy planes of neighboring clusters of data points.}
\end{quote}
\edit{In the following, we will see that only simple cases allow for this approach while the general case is more involved. First, we will discuss the creation of vertices for the simplified mesh (Section~\ref{sec:VerticesForASimplifiedMesh}). Subsequently, we will connect these vertices to faces in order to obtain the complete mesh (Section~\ref{sec:SimplificationFaces}). In both section we also address corresponding challenges as well as possible solutions.}

\subsection{Vertices for a Simplified Mesh}
\label{sec:VerticesForASimplifiedMesh}

\noindent \edit{The intuitive way to determine simplified mesh vertices is the intersection of neighboring proxies, as used by~\cite{lee2016feature}. In the following, we will use the notion of neighborhood for proxies as introduced at the end of Section~\ref{sec:UserControllerLevelOfDetail}.}

\paragraph{Intersecting Planes} 
\edit{A first naive solution for the creation of vertices for the simplified mesh is to consider the intersection of neighboring proxies and the construction of vertices in these intersection points. In the general case, where~$q>3$ proxy planes meet, we cannot simply consider the intersection as it will be mostly empty.
	
We call such a situation obtained via the proxy-neighborhood relation an \emph{$q$-tuple}. 
In detail, let~$I$ be the index set labeling the proxies given by subset~${P_i\subseteq P}$, proxy center~$C_i$, and proxy normal~$N_i$. 
Further, consider all neighborhood relations~${(i,j)\in I\times I}$ between proxies.
A~\emph{$q$-tuple} is a subset~${\{i_1,\ldots,i_q\}}$ of~$I$ such that the proxies~$i_j$ and~$i_\ell$ are neighbors for all~${i_j,i_\ell\in\{i_1,\ldots,i_q\}}$ where~${i_j\neq i_\ell}$.
Enumerate all such $q$-tuples~$I_j$ by an index set~$J$ in a way, s.t.\@ if there are~${j,j'\in J}$ with~${I_j\subsetneq I_{j'}}$, then the larger set~$I_{j'}$ is kept and~$j$ is not stored in~$J$.
These tuples hold inclusion maximal candidates for intersection points with proxy indices contributing to the intersection. 

Now, for the case of~$q>3$, we select for each tuple three indices at random, intersect them, and make the resulting vertex known to all proxy members of the tuple. 
As this might cause degenerate faces in the face creation stage---as the vertex does not lie within all of the proxies it is associated to---,we use a triangulation of all created faces (see Section~\ref{sec:SimplificationFaces}) to obtain triangles, which are planar.}


\paragraph{Intersecting Point Optimization}
\label{sec:VariationalTangentPlaneIntersection}

\noindent \edit{A second, more involved solution for the creation of simplified mesh vertices is based on optimization. 
The intersection of more than three planes is numerically unstable as discussed above. 
In the work of~\cite{zimmer2012variational}, the authors turn to a variational approach,} start from a triangle mesh, and aim at computing the intersection points~$x_j$ of the vertex tangent planes of all triangles. 
Thus, exactly three tangent planes---corresponding to each of the three vertices~$v_i$ of a triangle---contribute to an intersection point. 
Denoting the normal at~$v_i$ by~$n_i$, they solve the following minimization problem
\begin{align}
\label{equ:MinimizationZimmer}
\begin{split}
\text{minimize:}\quad & \sum_{t_j}\left(\sum_{v_i\in t_j}\left\|x_j-v_i\right\|_2^2\right)\\
\text{subject to:}\quad & n_i^T(v_i-x_j) = 0\quad \forall t_j,\quad \forall v_i\in t_j\\
& \left\|n_i\right\|_2^2=1\quad \forall v_i
\end{split}
\end{align}
where the normals are variables in the minimization. Note that the original normals at the vertices are not taken into account at all during the minimization. The requirement of unit-length normals is necessary, however, as otherwise~$n_i=0$ would trivially satisfy all conditions.

We generalize this approach in the following way to our setup. \edit{First of all, we use the concept of and notation for \emph{$ q $-tuples} introduced in the previous paragraph.} Then, we consider the following energy
\begin{align}
\label{equ:TangentialPlaneEnergy}
F(\{ x_1, \ldots, x_j \}) = \sum_{j \in J} \left( \sum_{i \in I_j} \left\| x_j - C_i \right\|_2^2 \right) + \sum_{i \in I} \tilde{w}_i \left\| N_i - \tilde{n}_i \right\|_2^2,
\end{align}
with sought-for intersection points~$x_j$ for each maximal tuple~$I_j$, known proxy-centers~$C_i$, weighting terms~$ {\tilde{w}_i \in \mathbb{R}_{\geq 0}} $, unknown normal deviations~$ \tilde{n}_i $, and known proxy-normals~$N_i$ for all proxies~${i\in I}$. Ultimately, we want to solve
\begin{align}
\label{equ:MinimizationTangentialPlane}
\begin{split}
\text{minimize:}\quad & F(\lbrace x_1, \ldots, x_j \rbrace) \\
\text{subject to:}\quad & \tilde{n}_i^T (x_j - C_i) = 0 \quad \forall j \in J \quad \forall i \in I_j\\
& \left\| \tilde{n}_i \right\|_2^2 = 1 \quad \forall i \in I.
\end{split}
\end{align}
This generalizes the problem of~\cite{zimmer2012variational} as stated in Equation~(\ref{equ:MinimizationZimmer}) in several ways. First, we allow for more than three, namely for an arbitrary number of planes to intersect. This arises already at a simple geometry like the octahedron, which has valence~$4$ vertices, see Figure~\ref{fig:zimmerIllustration} for an illustration. Second, we do allow the normals~$\tilde{n}_i$ to deviate from the proxy normals~$N_i$, but a large deviation is punished, where the severity can be steered by the weights~$\tilde{w}_i$.

\edit{In contrast to the first naive solution, this approach guarantees all vertices of the mesh to lie within the proxies that they are derived from. That is, if the optimization problem~(\ref{equ:MinimizationTangentialPlane}) yields a feasible point, after correcting the proxy normals from~$n_i$ to~$\tilde{n}_i$ all vertices associated to a planar proxy lie completely within the corrected proxy plane.}


\subsection{Faces for a Simplified Mesh}
\label{sec:SimplificationFaces}

\noindent \edit{After creating the mesh vertices, we need to connect them in order to generate face for the mesh. 
The general idea is to represent every proxy region with a single star-convex face. 
All vertices associated to a proxy are sorted around the barycenter of the proxy w.r.t.\@ an arbitrary reference direction. 
This yields correct results, when the barycenter of the proxy is also a star-convex center point, see Figure~\ref{fig:starConvex}. 
If we sort the vertices with a non-star-convex center, they are connected in wrong order and the resulting faces will degenerate. 
This approach obviously fails, if a proxy represents a non-convex part of the geometry, see Figure~\ref{fig:torusExample}.}

\begin{figure}
	\centering
	\captionsetup[subfigure]{justification=centering}
	\begin{subfigure}[t]{0.32\textwidth}
		\includegraphics[width=1.\textwidth]{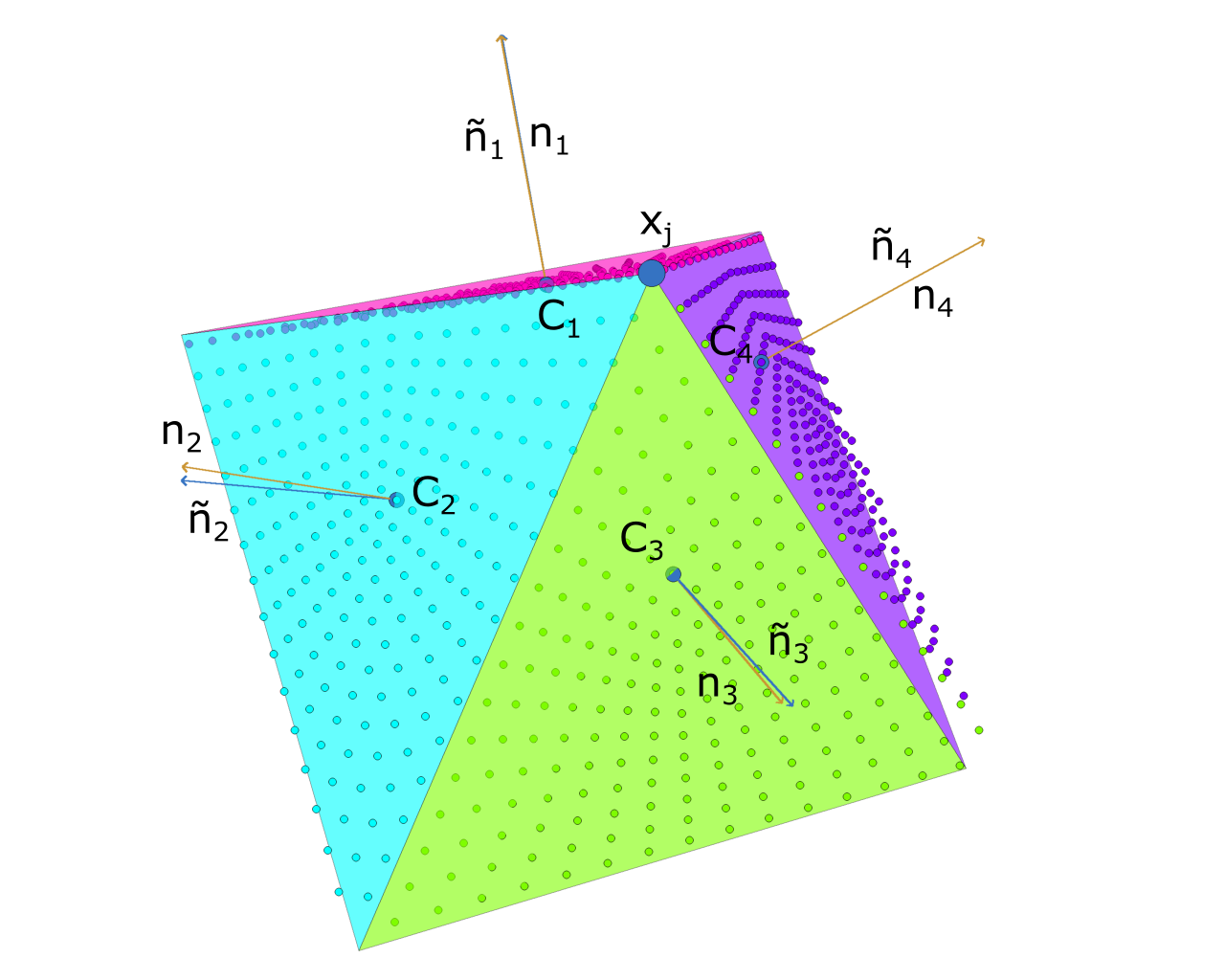}
		\caption{\edit{Intersection of more than three proxies}}
		\label{fig:zimmerIllustration}
	\end{subfigure}
	\hfill
	\begin{subfigure}[t]{0.32\textwidth}
		\includegraphics[width=1.\textwidth]{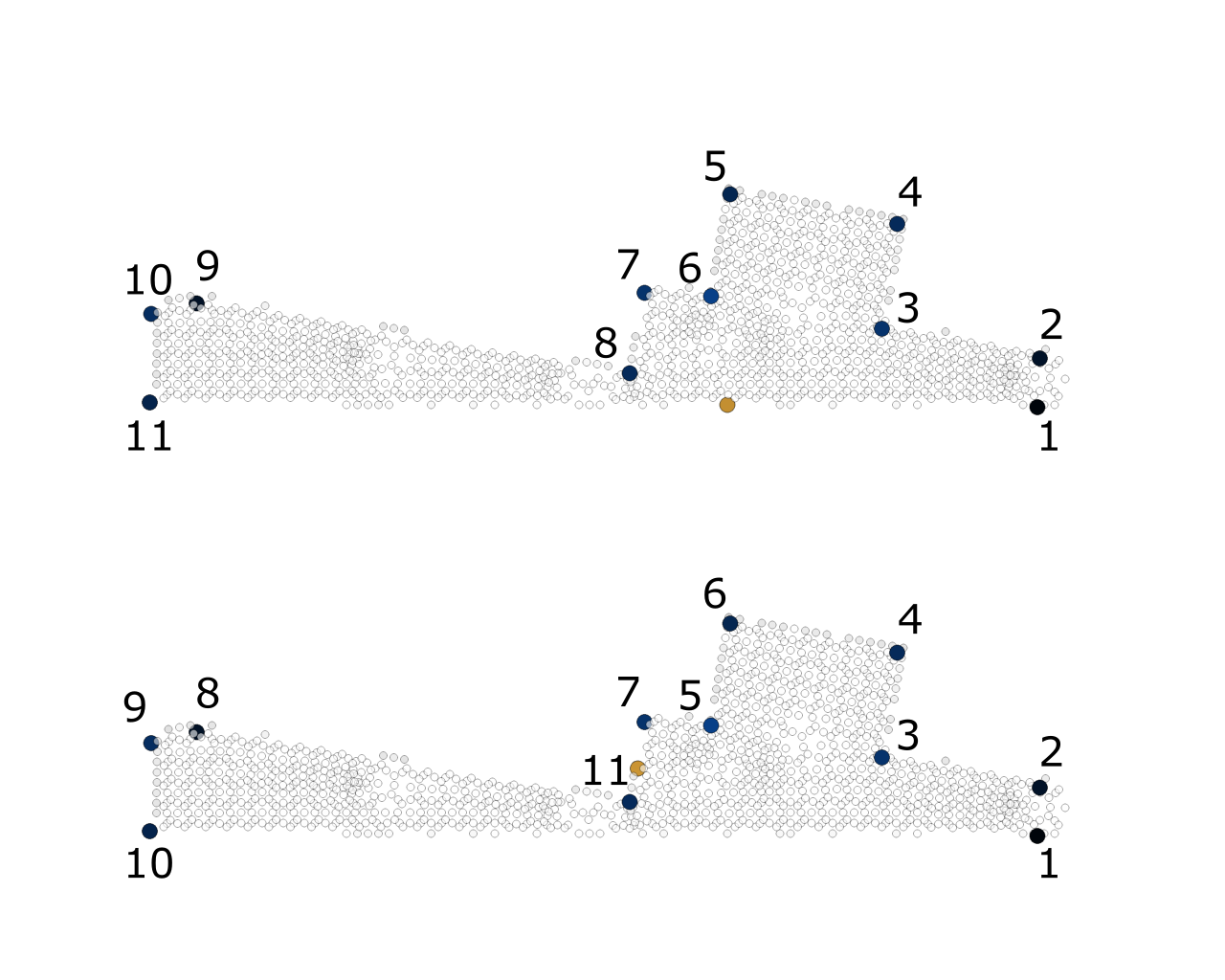}
		\caption{\edit{top: star-convex; bottom: barycenter}}
		\label{fig:starConvex}
	\end{subfigure}
	\hfill
	\begin{subfigure}[t]{0.32\textwidth}
		\includegraphics[width=1.\textwidth]{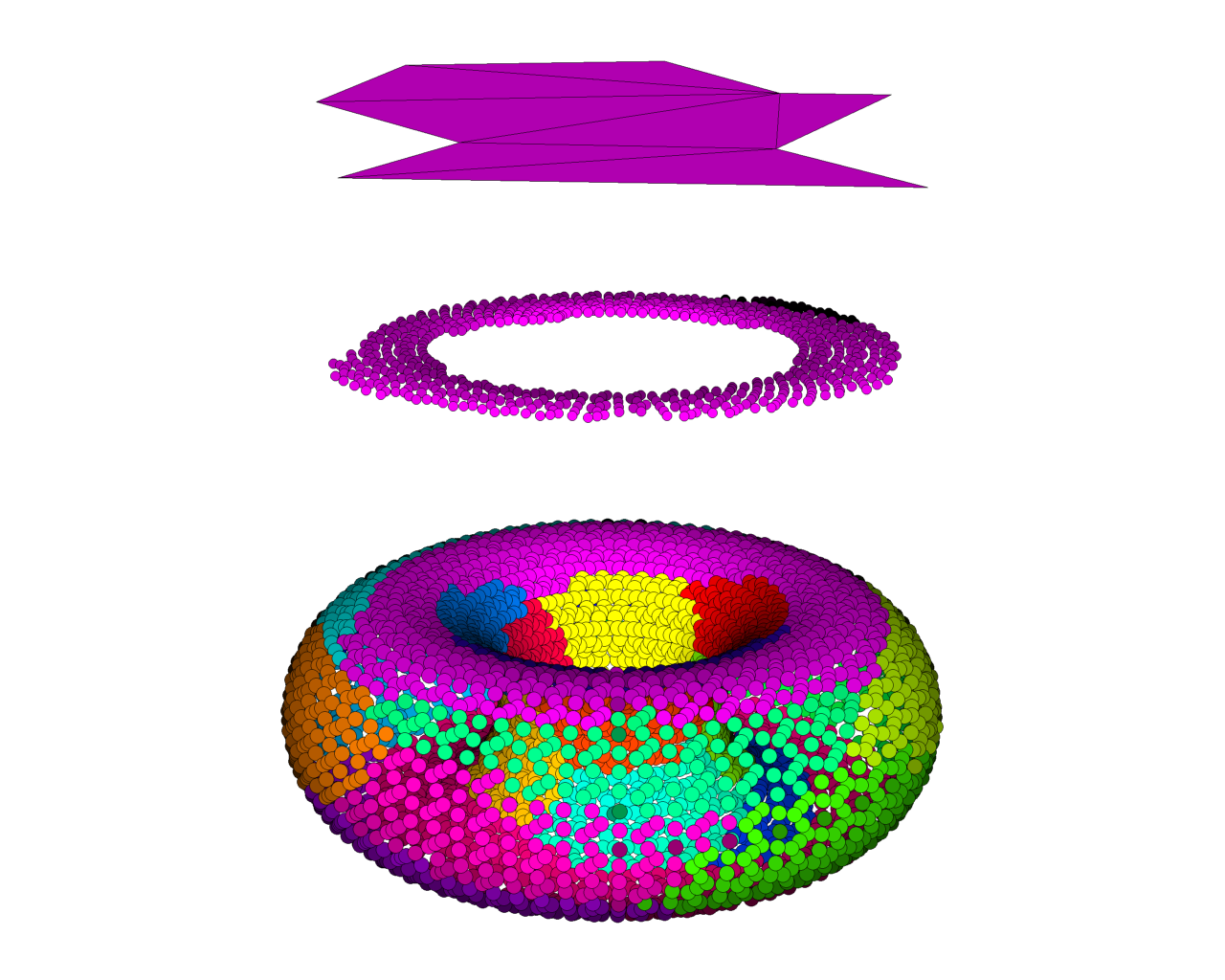}
		\caption{\edit{Reconstructing a non-convex proxy}}
		\label{fig:torusExample}
	\end{subfigure}
	\caption{\edit{(a) Illustration of point optimization with deviation allowance (represented by~$ \tilde{n} $) of $ n_1, \ldots, n_4 $ to find the optimal mesh vertex $ x_j $ satisfying the constraints $ \langle C_i - x_j, \tilde{n}_i \rangle = 0 $ and $ \left\| \tilde{n}_i \right\|_2^2 = 1 $ for $ i = 1, \ldots, 4 $. The $ C_i $ are the normal-corresponding proxy-centers and $ n_i $ the original proxy normals.
	(b) Star-convex proxy (front part of the fandisk, Figure~\ref{fig:fandisk}). Top shows an ordering of the vertices around a star-convex center point, bottom shows the corresponding ordering around the barycenter.\\ 
	(c) Bottom: Segmentation of the Torus with $ 24 $ final proxies. Middle: Non-convex proxy representing the upper part of the torus after projection of the points onto their proxy plane. Top: Resulting mesh face after connecting the vertices.}}
	\label{fig:sth}
\end{figure}


\section{Experimental Results}
\label{sec:results}

\edit{\noindent The following experimental section is divided into two parts. First, we evaluate parameter choices regarding the segmentation and provide a quantitative comparison of the segmentation results. In the second part, we discuss three aspects of the simplification algorithm on clean and noisy models.
	
In all our experiments, we processed models with quite uniform samplings. Hence, for simplicity,  we utilized equal weights~$ {\omega_{j} = 1} $ in Equation~(\ref{equ:L21energyPointsOneProxy}). We proceed similarly with the weight assignment in the optimization problem formulated in Equation~(\ref{equ:TangentialPlaneEnergy}) and set~$ {\tilde{w} = 1} $. 
We use neighborhoods to both propagate a proxy during the \emph{flood} step and establish neighborhood relations between the different proxies.
For the first purpose, we use a combinatorial $k$-nearest neighbors approach. 
When determining the proxy neighborhoods, we turn to a combination of the combinatorial and geometric approach based on the same~$ k $ (see end of Section~\ref{sec:UserControllerLevelOfDetail}). 
In all our experiments, we use $ k = 8 $. Deviations from the default parameters are indicated.}

\subsection{Segmentation}
\label{sec:SegmentationResults}

\edit{
\noindent For a large scale experiment, we chose~$ 600 $ models from the repository used in the work~\cite{hu2018tetrahedral}. For all these models, we used the mesh information to generate an oriented vertex normal field. Furthermore, we translated the models and scaled them uniformly to fit into the unit cube. Finally, we performed our experiments on the point cloud given by the mesh vertices, disregarding the connectivity information and the triangular faces.
	
We compared four different approaches. The first one was the segmentation algorithm of~\cite{rabbani2006segmentation} as implemented in the \emph{Point Cloud Library} of~\cite{rusu20113d}. As parameters, we turned to the ones described in the original paper, see~\cite{rabbani2006segmentation}. We will refer to this experiment by \expA. Second, with these results at hand, we took the final numbers of proxies given by \expA\ for each geometry. This number served as number of proxies to be sought by the variational shape approximation algorithm of~\cite{lee2016feature}. Here, no splits or merges are applied, thus we refer to this experiment by~\expB. Third, we took the total~$ \mathcal{L}^{2,1} $-error (Equation~(\ref{equ:L21energyPoints})) of each geometry, as produced by the \expA\ experiment, and divided it by its final number of proxies. This division provides an initial guess for a local, geometry-dependent value~$ \eta $. In this third experiment, we allowed splits as well as merges. Also, we started with an initial seed number of~$ {m=\aleph} $, i.e.\@ each point was a seed at the start. Because of the $ \eta $-threshold and merge-processes, the number of proxies reduced drastically over the run of the experiment. We will refer to this as \expC. Fourth and finally, without any priors, we set~$ {\eta = 25} $ and allowed splits and merges. Furthermore, we once more started with every point as a seed. The choice of~$\eta$ is motivated from previous experiments. We will refer to this fourth experiment as \expD. The terminology \emph{local} or \emph{global} indicates whether~$ \eta $ is chosen with respect to the geometry or globally for all geometries. Observe that the experiments~\expB\ and \expC\ are dependent on the results of \expA, while only \expD\ is independent.

We are interested in gaining insight into the relationship between the obtained proxy-number~$m$ and the quality of the induced flat proxy-regions. Besides the~$ \mathcal{L}^{2,1} $-measure of Equation~\ref{equ:L21energyPoints}, we focus on the mean squared error (MSE) caused by point-to-proxy-plane distances to evaluate the region quality. The MSE is given as
\begin{equation}
	MSE(\lbrace(P_i, N_i) \mid i = 1, \ldots, m\rbrace) = \frac{1}{\aleph} \sum_{p_j \in P_i} \left\|p_j-\pi(p_j)\right\|_2^2,
\end{equation}
where $ \pi(p_j) $ denotes the orthogonal projection of $ p_j $ onto its related proxy plane, given by normal~$ N_i $ and base point~$ C_i $.}

\edit{From the~$ 600 $ chosen models, we obtained~$ 499 $ that offered segmentation results in all four experiments. For~$27$ models, \expA\ was unable to provide a valid segmentation, because it assigned a zero proxy-normal to at least one region (for instance a region holding only two antipodal normals). These models were excluded for the subsequent experiments. The variational shape segmentation of~\cite{lee2016feature} did not report a complete segmentation on additional~$72$ models. Here, some points are not assigned to any proxies, because they cannot be reached from the proxy centers during a \emph{flood} when traversing the nearest neighbor graph. Increasing the parameter~$k$ alleviates this problem. Similar failures occurred on one additional model in experiments \expC\ and \expD\ respectively. Even though these experiments started with seed numbers equal to the geometries' points, they reduced the number of regions via \emph{merge} operations. Due to proxy updates and new seed selection, it is possible that seeds travel away from sparsely sampled areas, where they do not reach all formerly assigned points in the neighborhood graph during the next \emph{flood}. This reduced the number of models by a total of~$101$ failures to~$499$ feasible models. All reported experimental values are taken over this set of~$499$ models.

\begin{table}
	\centering
	\begin{tabular}{r||r r r r}
		& \expA & \expB & \expC & \expD \\
		\hline \hline
		min MSE & 2.75E-09 & 7.31E-07 & \textbf{6.95E-11} & 2.11E-06
		\\
		max MSE & 1.12E-01 & 2.98E-02 & 4.15E-02 & 1.04E-02\\
		avg MSE & 1.59E-02 & \textbf{2.88E-04} & \textbf{6.40E-04} & \textbf{4.50E-04}\\
		sd MSE & 2.12E-02 & 1.90E-03 & 2.47E-03 & 8.15E-04 \\
		\hline \hline
		min m & 1.00 & 1.00 & 2.00 & 1.00 \\
		max m & 1,103.00 & 1,103.00 & 393.00 & 135.00 \\
		avg m & 174.17 & 174.17 & \textbf{55.05} & \textbf{33.20} \\
		sd m & 163.10 & 163.10 & 58.41 & 20.69\\
	\end{tabular}
	\caption{Statistical evaluation of error MSE and proxy number m taken over all $ 499 $ geometries. We give the minimum, maximum, mean, and standard deviation.}
	\label{tab:errors}
\end{table}

For the following analysis, we turn to Table~\ref{tab:errors}. There, we give statistics on both the MSE as obtained from experiments on our model set. Regarding the average MSE over all experiments, we see that all three experiments---\expB, \expC, and \expD---outperform \expA\ by two orders of magnitude. A direct comparison between the MSE obtained for the models reveals that \expC\ and \expD\ outperform \expB\ in roughly~$8.5\%$ of all experiments. Note that the minimal MSE error obtained over all geometries is up to five orders of magnitude smaller for \expC\ when compared with the other experimental setups.

Aside from the MSE results, Table~\ref{tab:errors} also reports statistics on the number of proxies obtained by the different experiments over all geometries. Recall that we are not only interested in small error values, but also in representations of the geometry that reduce its complexity, i.e.\ that have a low number of proxies. Towards this end, it is remarkable that \expC\ and \expD\ attain MSE values comparable to those of \expB\ while only using~$31.6\%$ and~$19.1\%$ of the proxies on average, respectively. The close error results are especially of interest for \expD, as this runs without any dependency or information provided by \expA, as a global parameter is applied to all geometries equally. Hence, the assignment of points to proxies is on one hand optimized in terms of the MSE errors measure, while providing significantly fewer proxies on the other hand. Note that the lowest (and therefore optimal) MSE of $0$ is given for a segmentation in which every point is represented by its own proxy.
}

\begin{figure}
	\includegraphics[width=1.\textwidth]{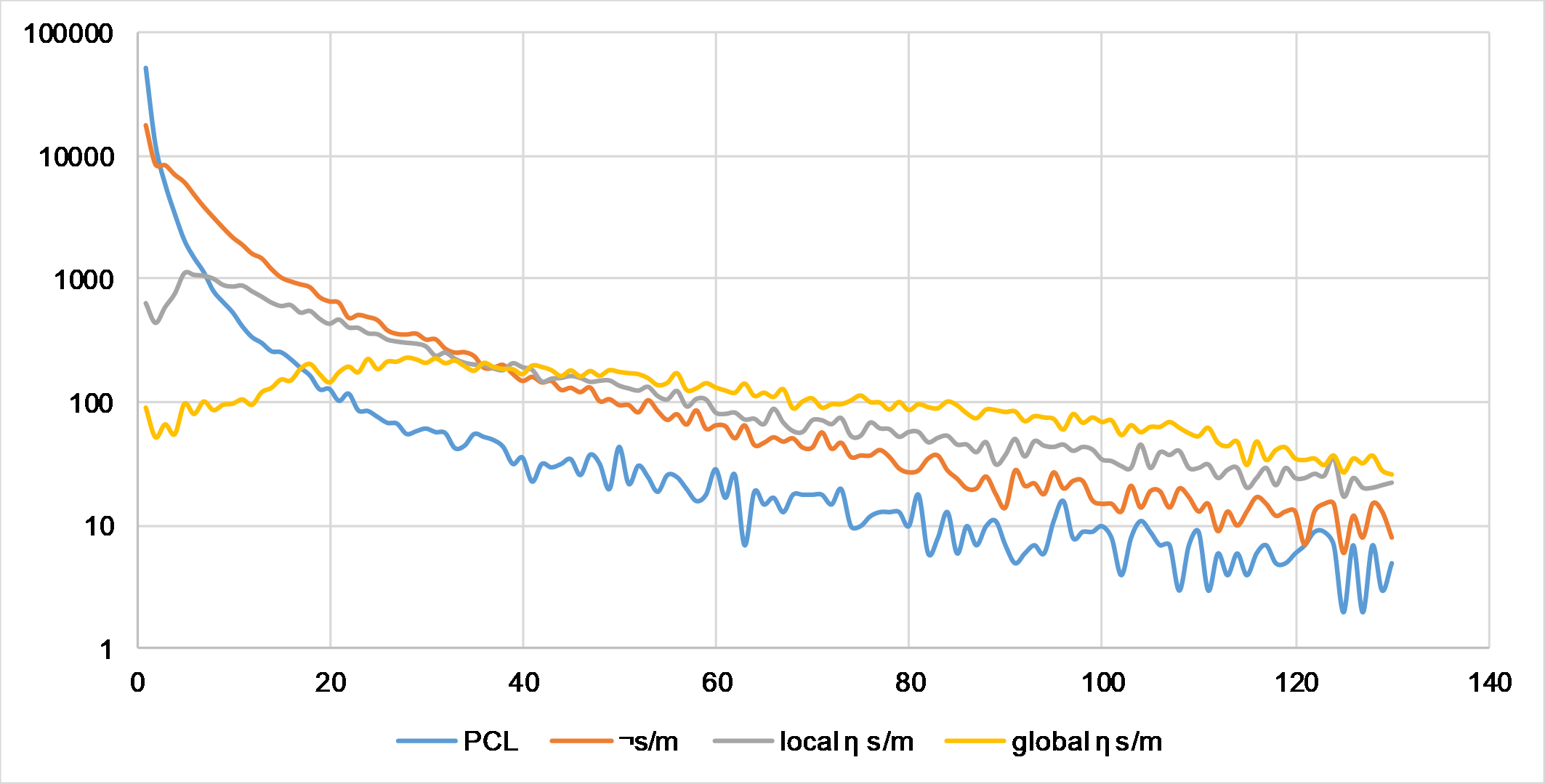}
	\caption{Histogram over all proxy-sizes up to $131$ among all $ 499 $ geometries for all four experiments. The upper bound of~$131$ is given by the sum of the mean ($11.74$) and corresponding standard deviation ($120.27$) regarding cluster sizes obtained from \expA. Note the logarithmic scale on the \mbox{$y$-axis}.}
	\label{fig:hist_ptsPerProxy}
\end{figure}

\edit{We proceed to further investigate the proxies obtained by the different experiments. In Figure~\ref{fig:hist_ptsPerProxy}, we show a histogram over the attained proxy-sizes taken over all geometries in the experiment. Note that the \mbox{$y$-axis} has a logarithmic scale. We show all proxy sizes up to~$131$, where this bound is given by the sum of the mean ($11.74$) and corresponding standard deviation ($120.27$) regarding proxy sizes obtained from \expA. We can see that both \expA\ and \expB\ create a significantly larger number of small proxies when compared with the segmentation results of \expC \ and \expD. In fact the average proxy sizes are~$11.74$ (\expA,\expB), $37.15$ (\expC), and $64.00$ (\expD). As segmentation---in our setup---should create few regions that still reflect the geometry attributes, extremely small regions as exposed by \expA\ and \expB\ are undesirable. The availability of splits and merges in \expC\ and \expD\ results in a bell-curve-like behavior in Figure~\ref{fig:hist_ptsPerProxy}, as both curves first increase and show a small descent with minor oscillations after their peaks. Hence, in the critical area of small sized proxies, the availability of splits and merges not only reduces their required number, but also balances their sizes, causing for more uniformly sized proxies.

To summarize the quantitative analysis of the segmentation part, we conclude:
\begin{itemize}
	\item The proposed method outperforms the segmentation approach of~\cite{rabbani2006segmentation} as well as variational shape approximation without splits and merges, as used by~\cite{lee2016feature} in regard of MSE.
	\item Without any knowledge of seed numbers or error values, a globally set~$ \eta $, availability of splits as well as merges, and treatment of all points as initial seeds provides segmentation results that have MSE comparable to~\cite{rabbani2006segmentation,lee2016feature} but a significantly reduced number of proxies.
	\item The availability of splits and merges not only optimizes for small proxy numbers, but also causes more uniform region sizes.
\end{itemize}} 

\subsection{Simplification}
\label{sec:SimplificationExperiments}

\edit{\noindent In the following, we present different experiments regarding the simplification as obtained from the proxy segmentation.
Each experiment addresses different aspects of the simplification pipeline.
First, we consider how the parameter~$\eta$ influences the obtained simplification (Section~\ref{sec:ThresholdEta}). 
Next, we turn to the Fandisk model, to discuss difficulties arising due to face generation (Section~\ref{sec:FandiskModel}). 
The last experiment deals with a noisy geometry and robustness of our algorithm (Section~\ref{sec:NoisyModel}).}

\edit{Throughout our experiments, in order to solve the minimization problem~(\ref{equ:MinimizationTangentialPlane}), we turn to the build-in solver of Matlab. Note that the minimization problem has a non-linear target function with non-linear constraints and can thus not be solved by any LP or even ILP solver. Hence, we follow the example from the Matlab manual~\cite{matlab2019solved}, which is supported by all versions newer than~\emph{R2019b}. The solver asks for starting points from which to run the optimization. We initialize the normals~$\tilde{n}_i$ by the proxy normals~$N_i$. As first guesses for any intersection point~$x_j$, we chose the barycenter of the centers~$C_i$ of those proxies that contribute to this intersection.}

\begin{figure}
	\centering
	\captionsetup[subfigure]{justification=centering}
	\begin{subfigure}[t]{0.25\textwidth}
		\includegraphics[width=1.\textwidth]{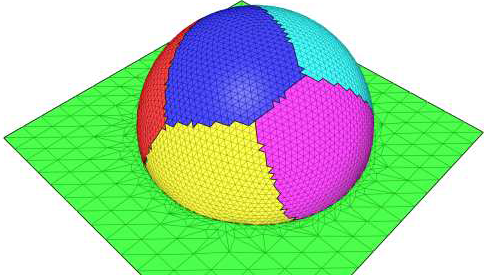}
		\caption{}
	\end{subfigure}
	\hfill
	\begin{subfigure}[t]{0.3\textwidth}
		\includegraphics[width=1.\textwidth]{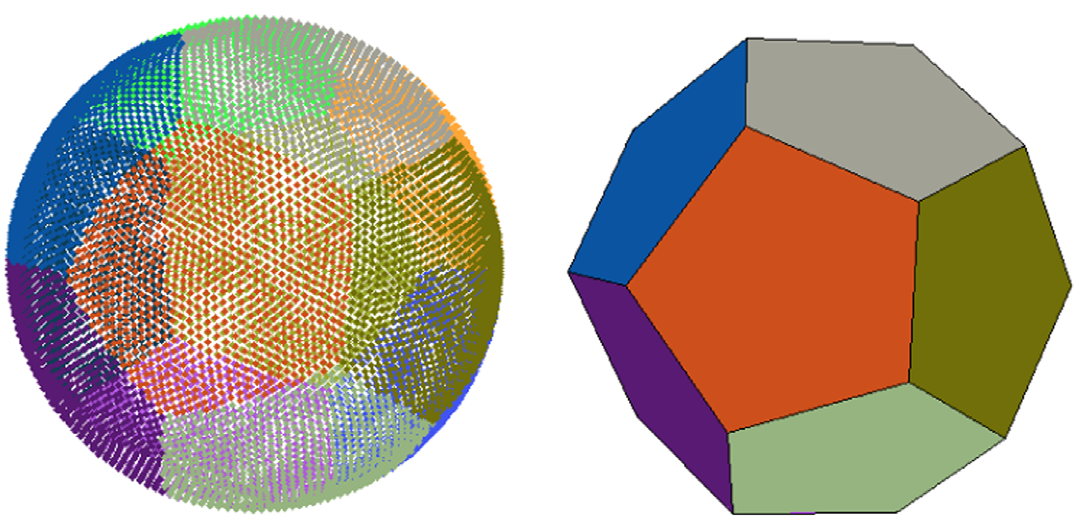}
		\caption{}
	\end{subfigure}
	\hfill
	\begin{subfigure}[t]{0.35\textwidth}
		\includegraphics[width=0.48\textwidth]{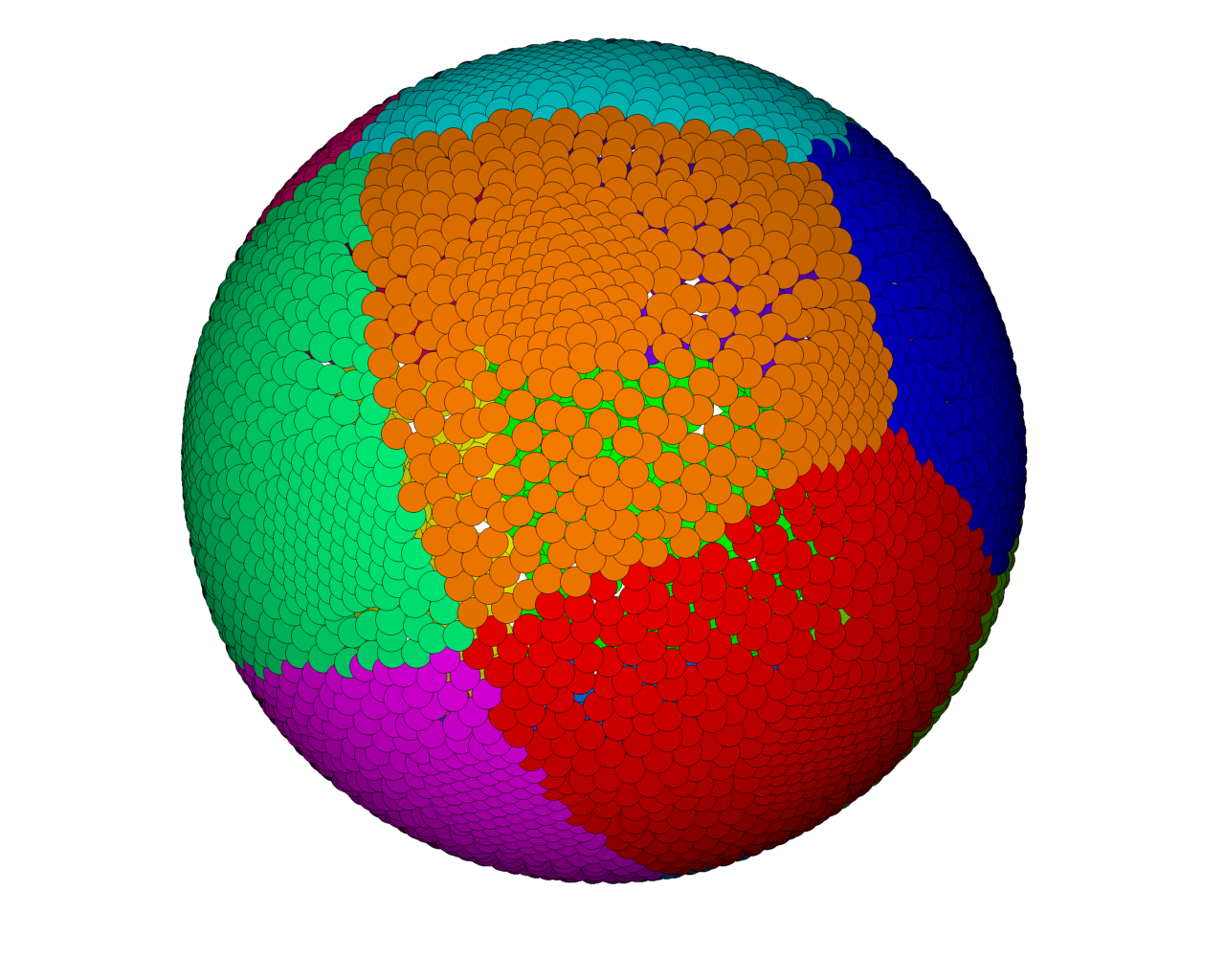}
		\includegraphics[width=0.48\textwidth]{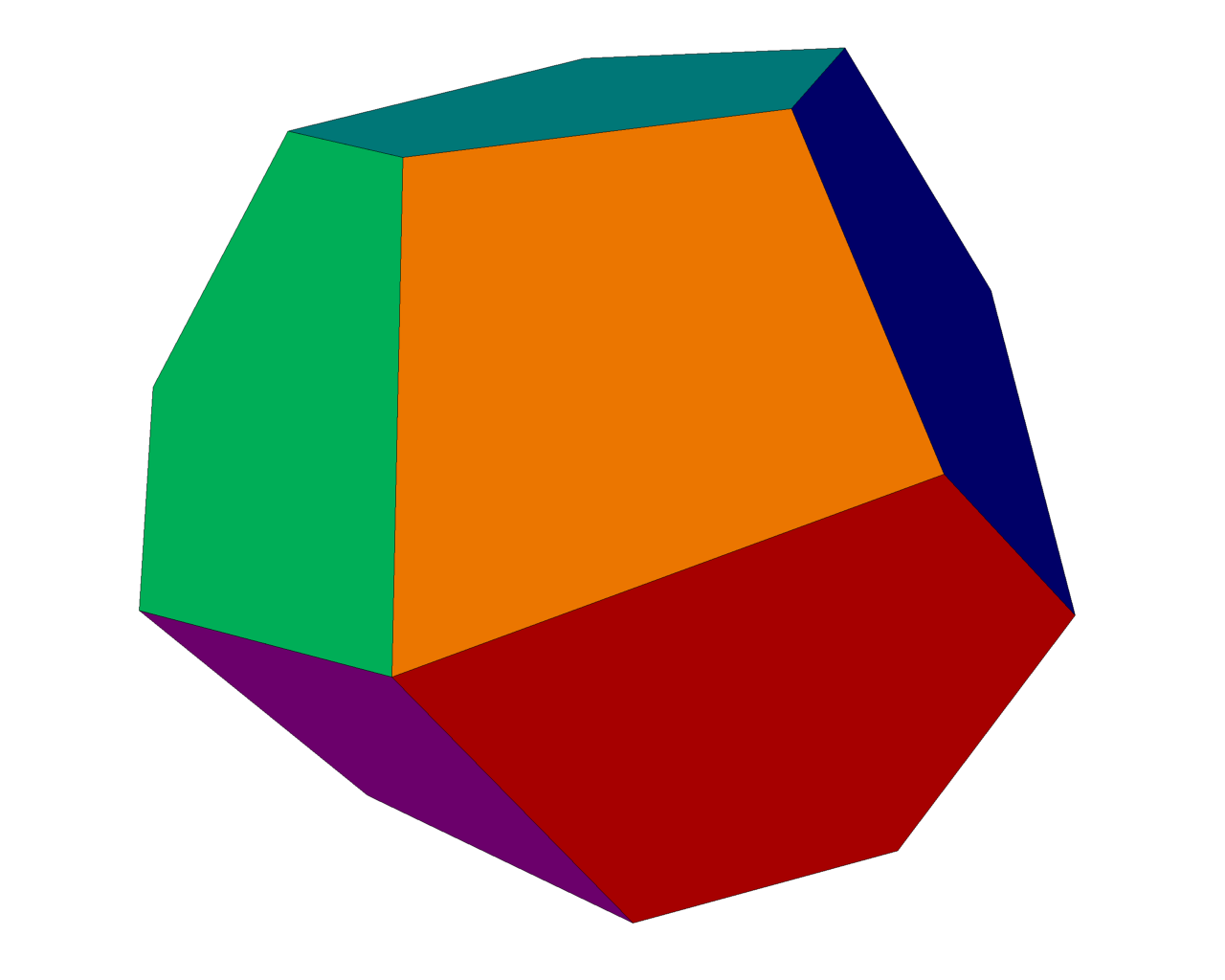}
		\caption{}
	\end{subfigure}
	\caption{A visual comparison of the output of (a)~\cite{cohen2004variational} showing a segmentation of the half-sphere into six proxies, (b)~\cite{lee2016feature} with a segmentation of the sphere into 12 proxies, and (c)~the results of our algorithm applied to the sphere deducing 12 proxies.}
	\label{fig:sphere}
\end{figure}

\subsubsection{Threshold $ \eta $-dependency on the Sphere Model}
\label{sec:ThresholdEta}
\noindent Our first simplification model is a sphere sampled with~${\aleph=5,122}$ points. We chose this model as it also appears in~\cite{cohen2004variational,lee2016feature}. By running our algorithm with~$12$ initial centers without split and merge we obtain a segmentation into~$12$ planar faces, shown together with the simplification done by optimization in Figure~\ref{fig:sphere} coupled with results of~\cite{cohen2004variational,lee2016feature}. 

\begin{figure}
	\centering
	\begin{subfigure}{0.18\textwidth}
		\includegraphics[width=1.\textwidth]{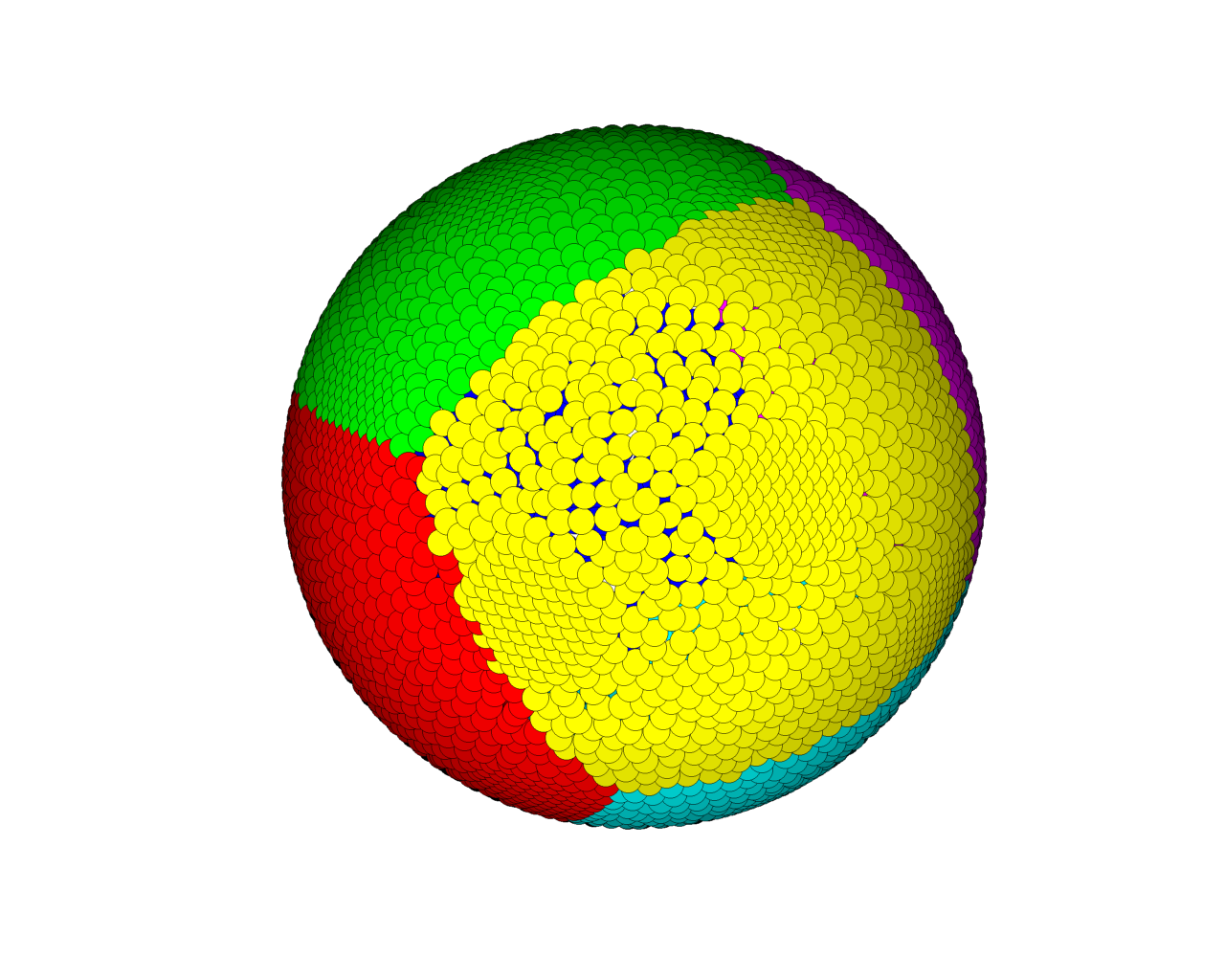}
		\includegraphics[width=1.\textwidth]{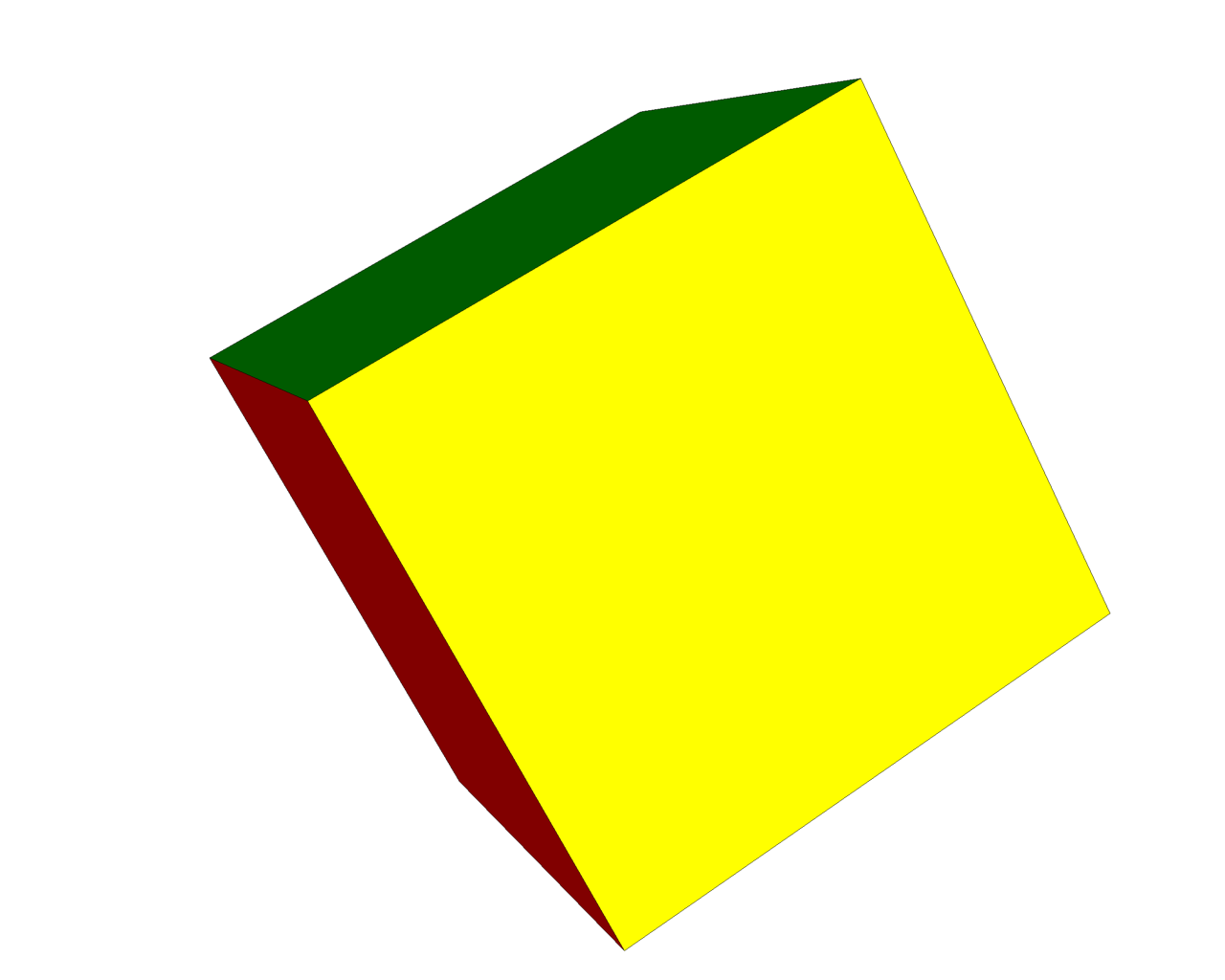}
		\includegraphics[width=1.\textwidth]{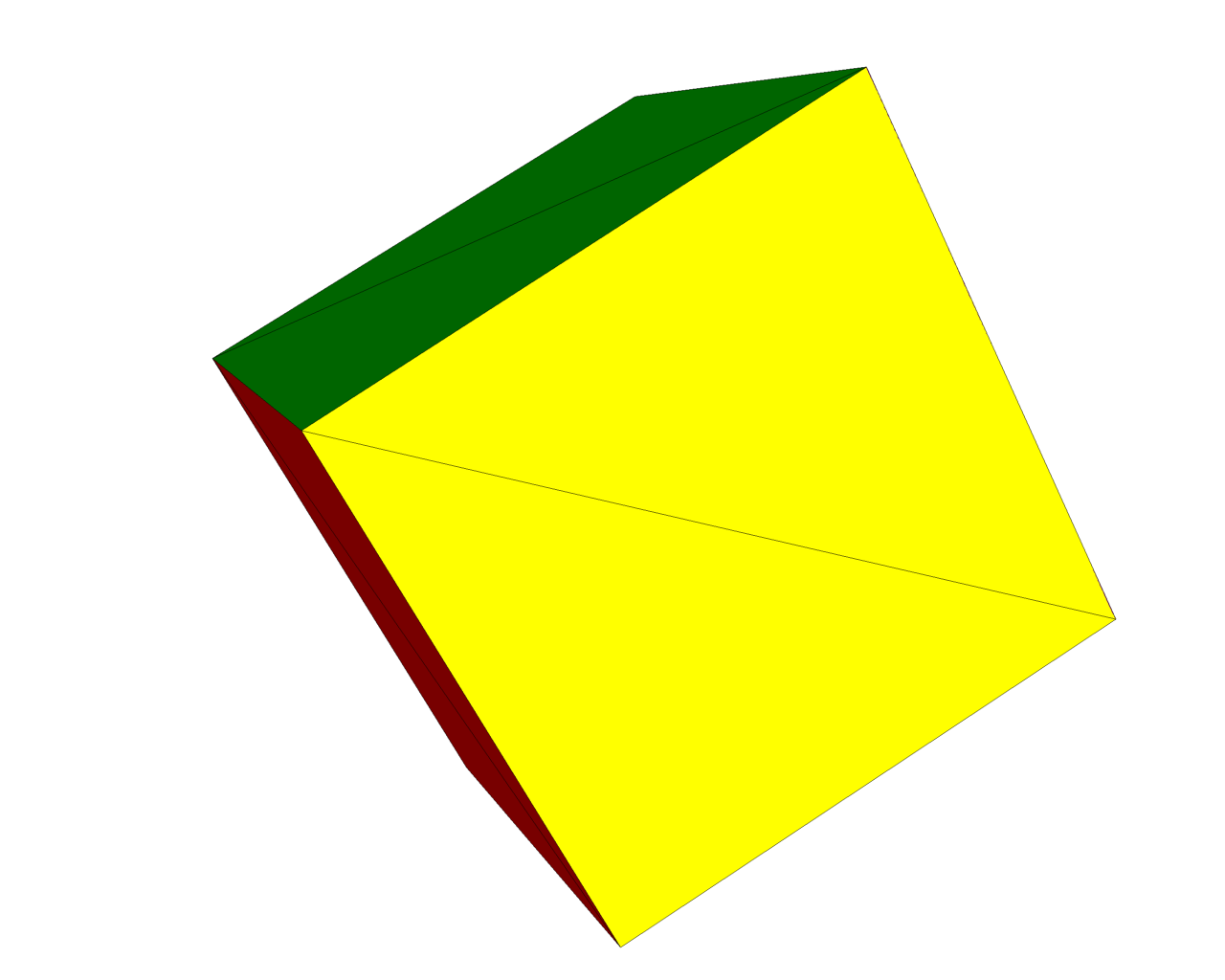}
		\caption*{$ \eta = 500 $\\ $ m = 6$}
	\end{subfigure}
	\begin{subfigure}{0.18\textwidth}
		\includegraphics[width=1.\textwidth]{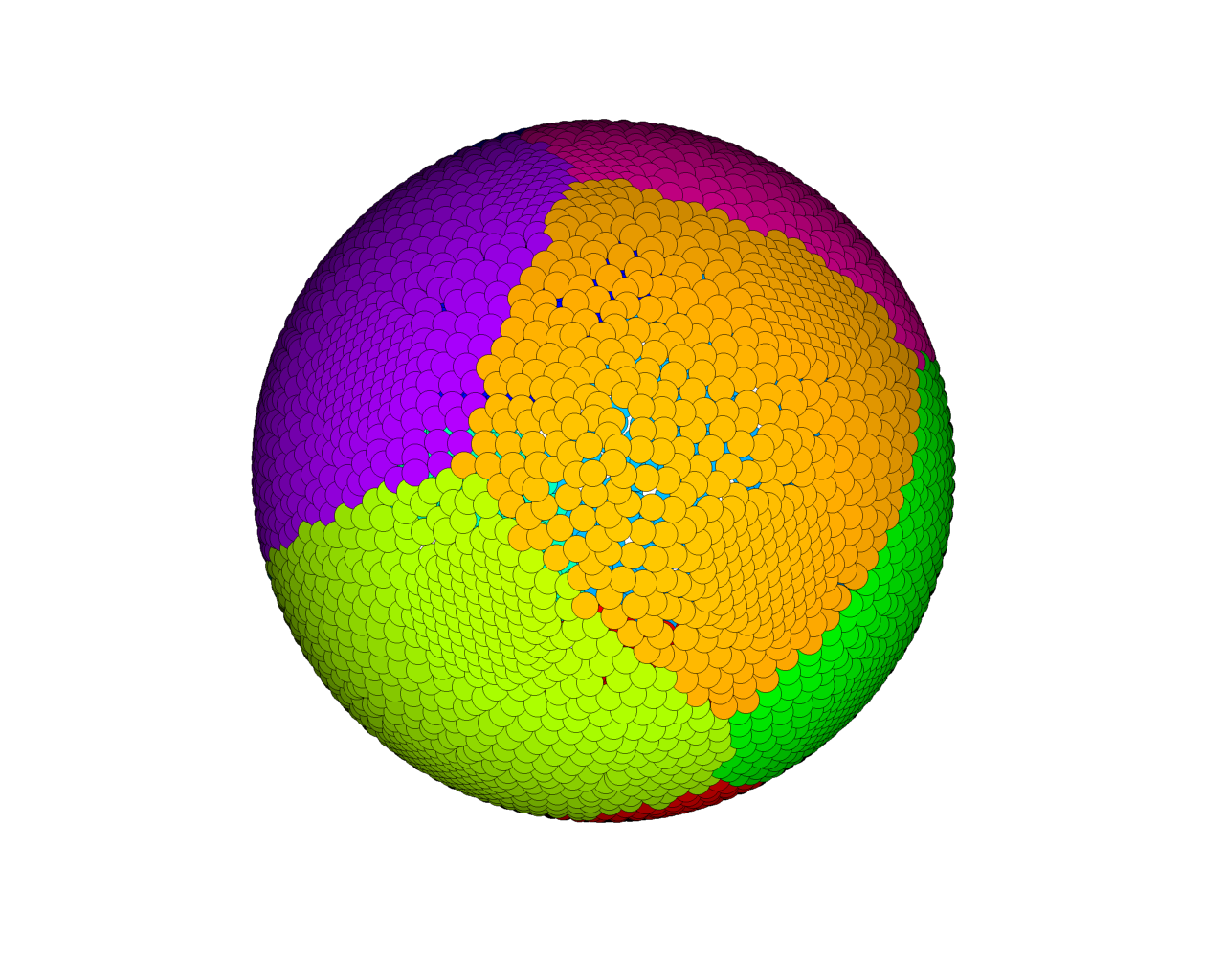}
		\includegraphics[width=1.\textwidth]{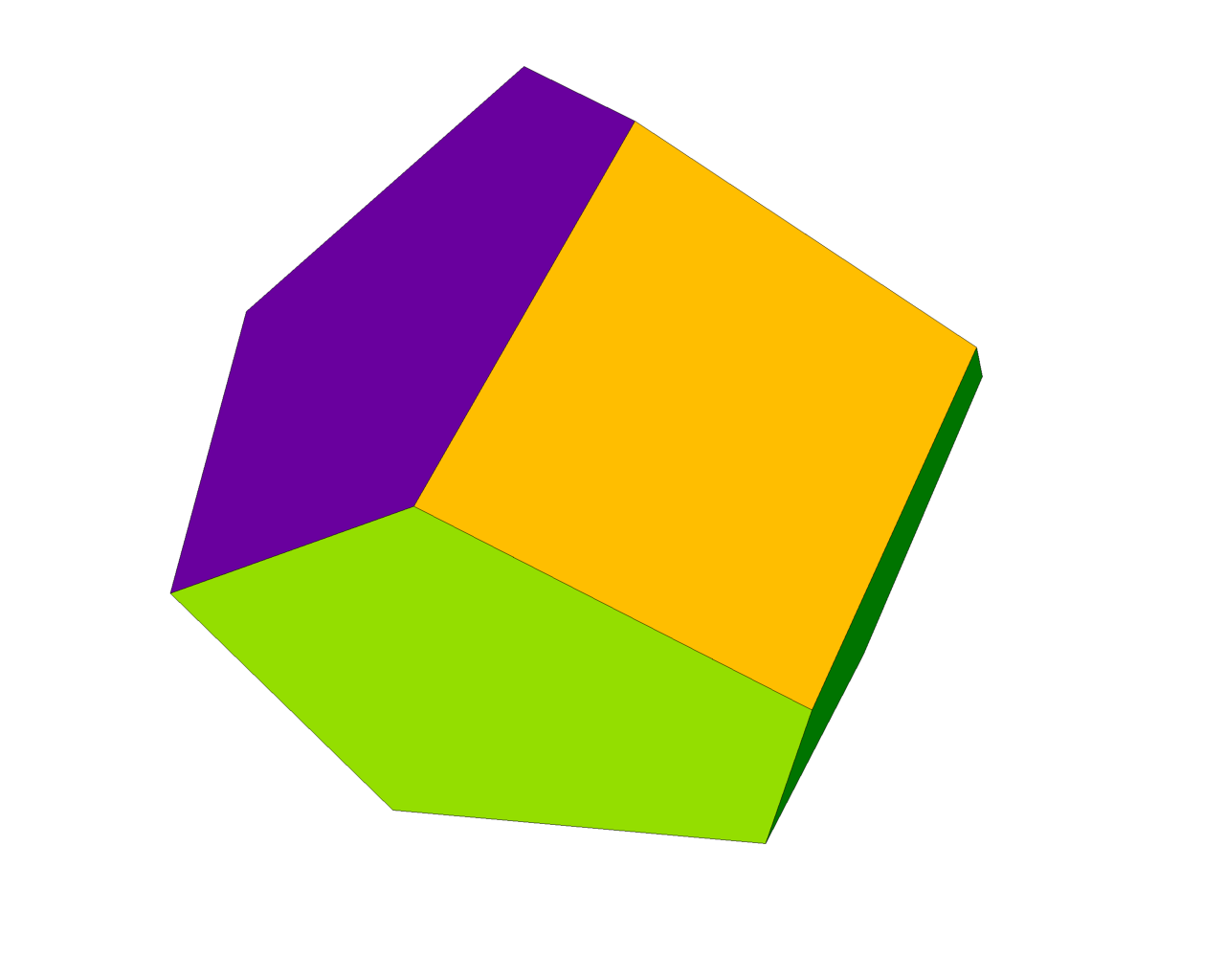}
		\includegraphics[width=1.\textwidth]{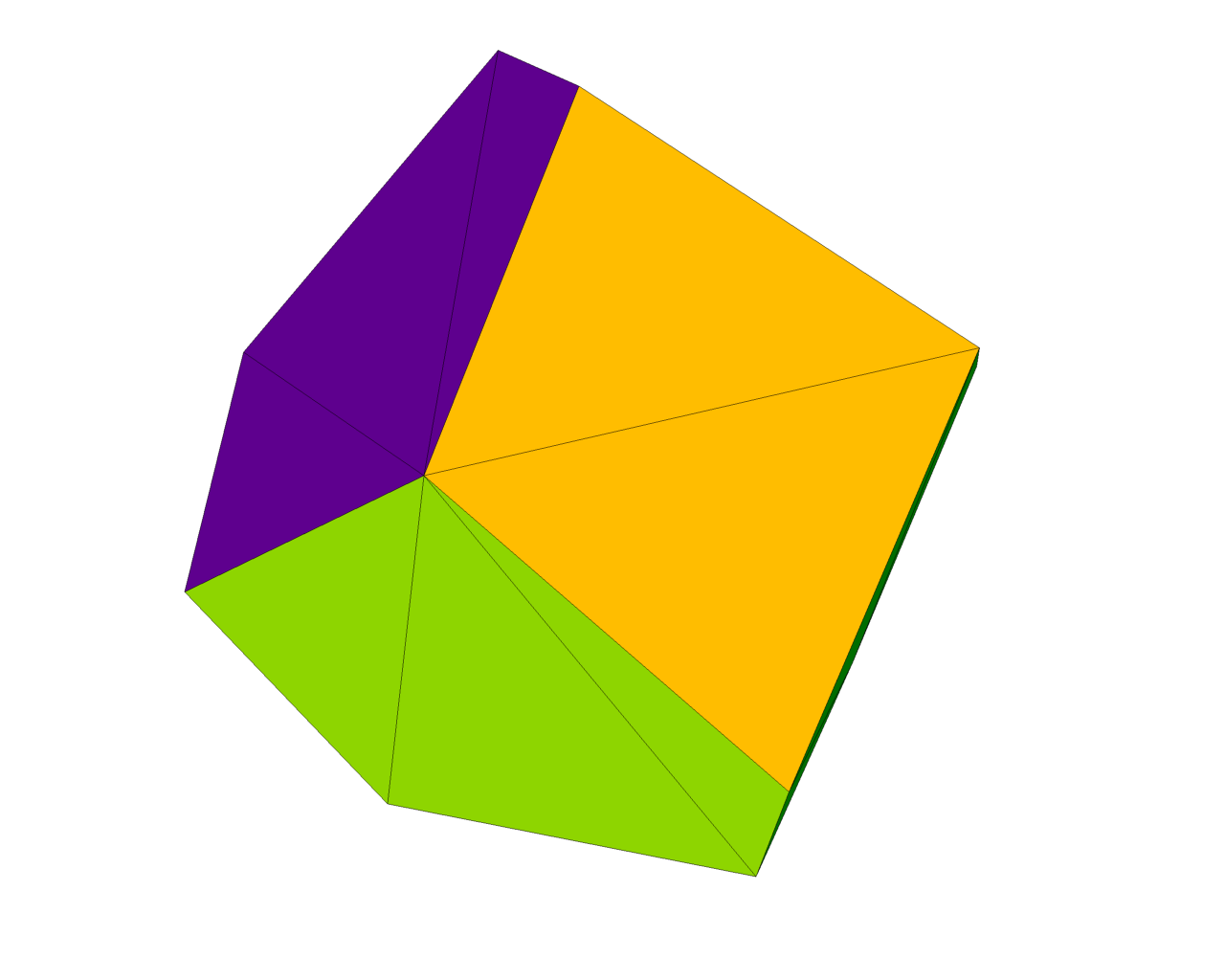}
		\caption*{$ \eta = 200 $\\ $ m = 9$}
	\end{subfigure}
	\begin{subfigure}{0.18\textwidth}
		\includegraphics[width=1.\textwidth]{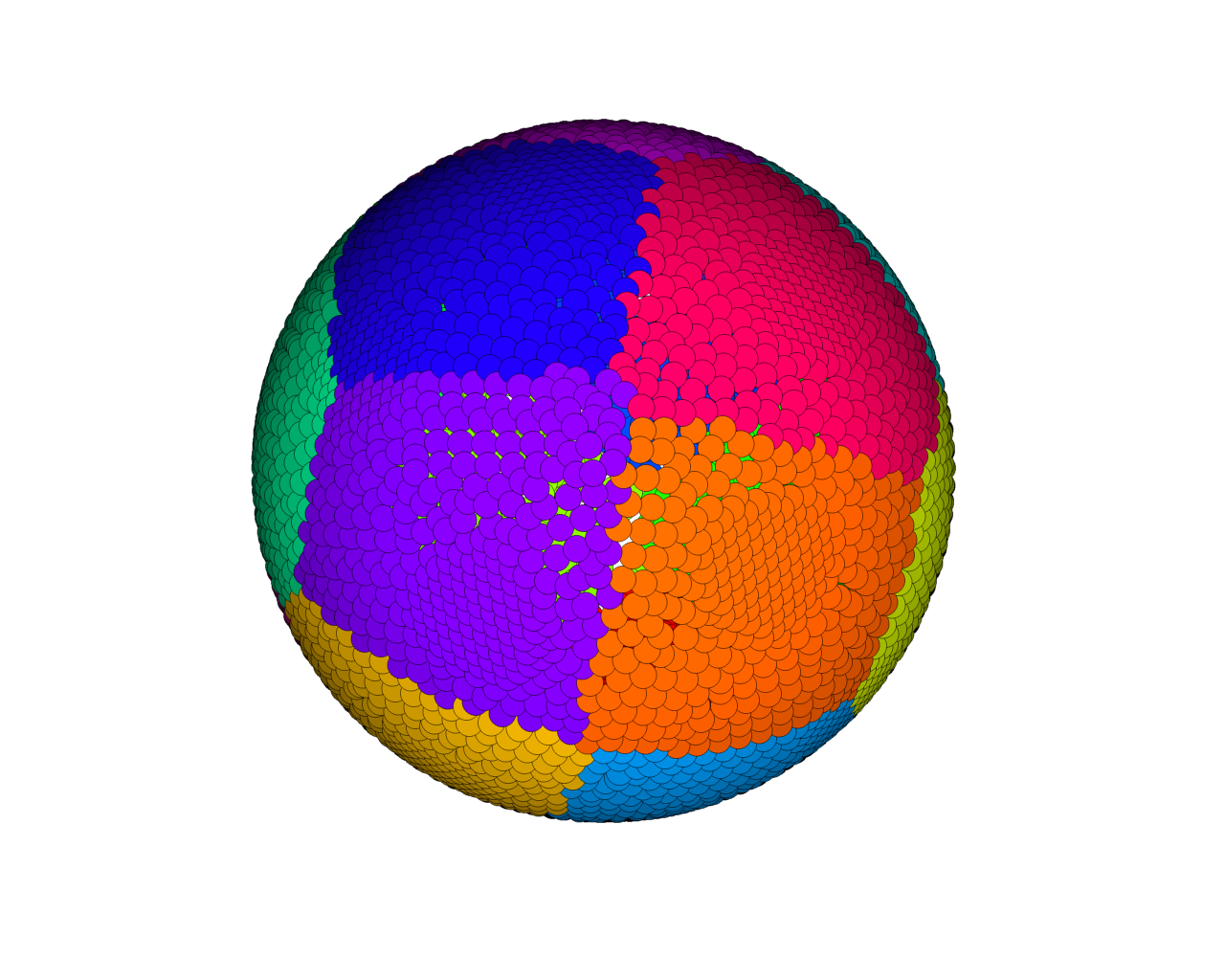}
		\includegraphics[width=1.\textwidth]{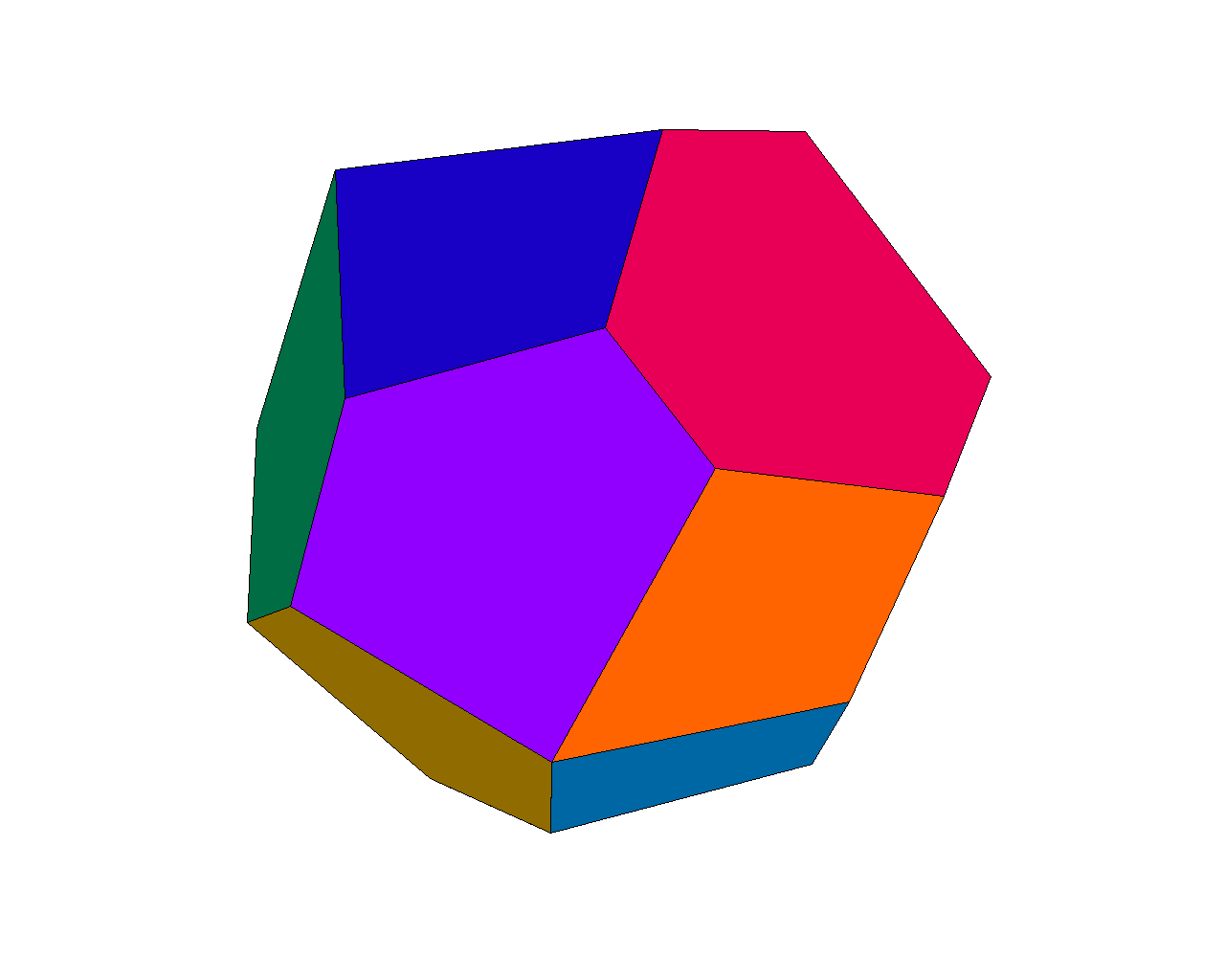}
		\includegraphics[width=1.\textwidth]{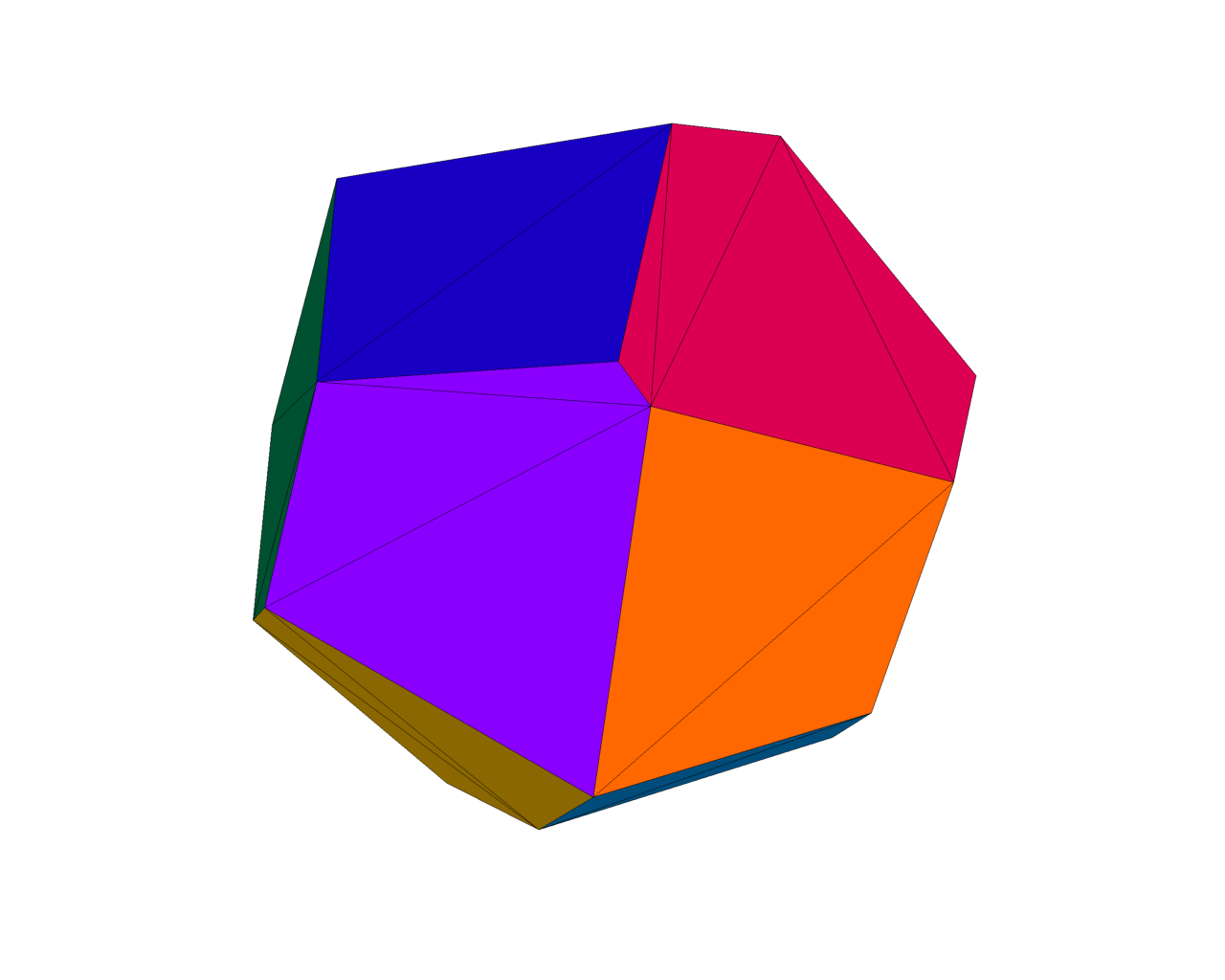}
		\caption*{$ \eta = 100 $\\ $ m = 16$}
	\end{subfigure}	
	\begin{subfigure}{0.18\textwidth}
		\includegraphics[width=1.\textwidth]{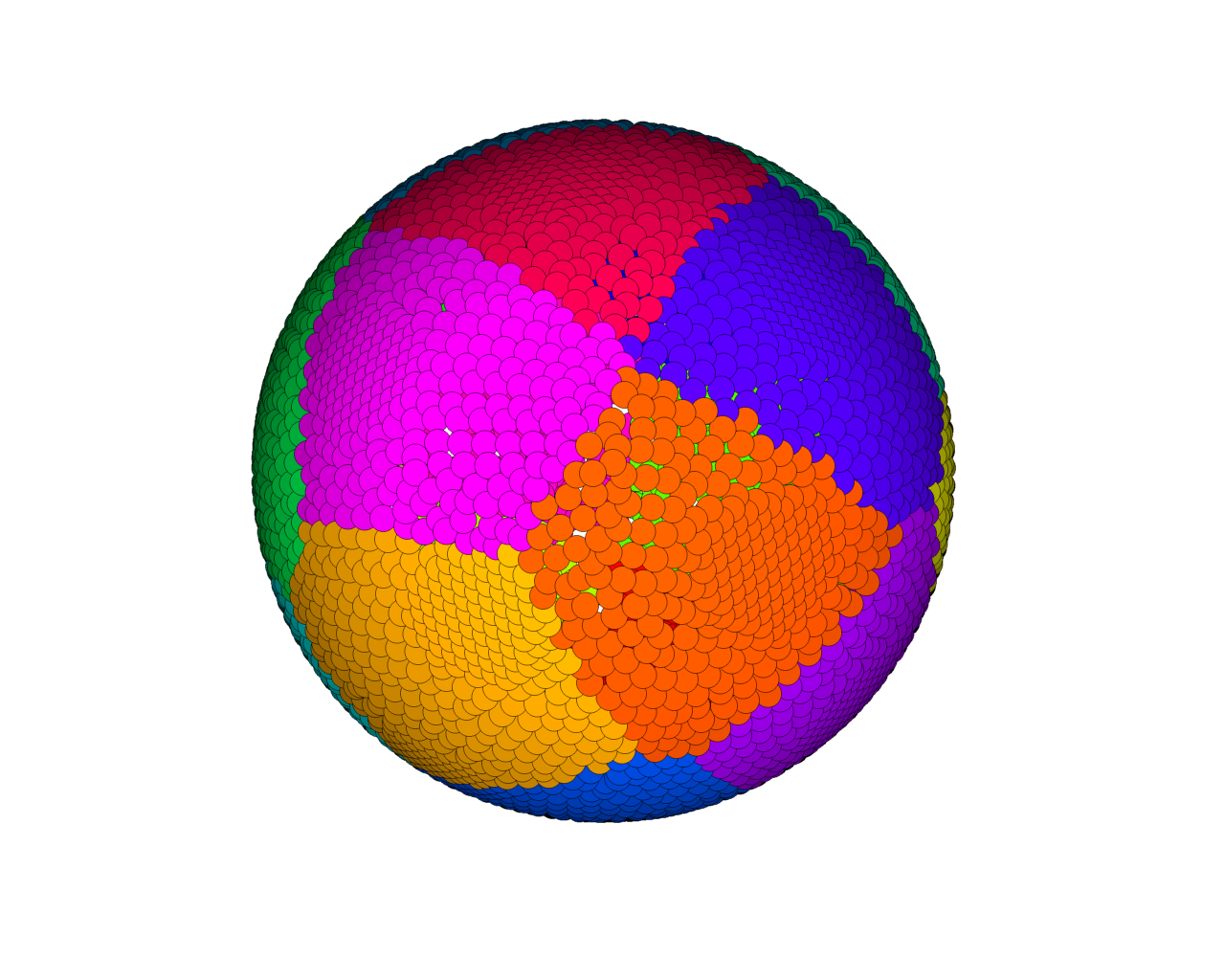}
		\includegraphics[width=1.\textwidth]{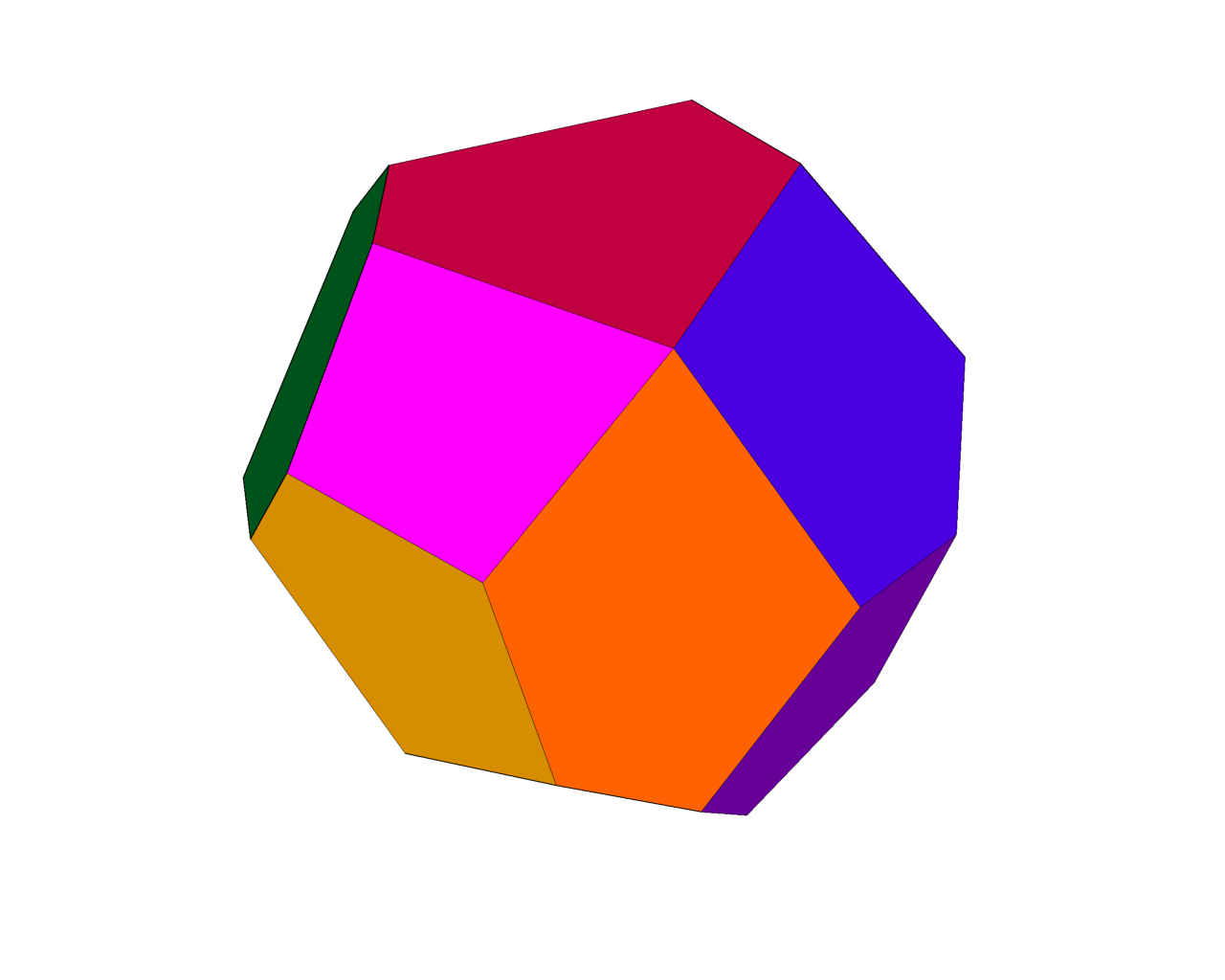}
		\includegraphics[width=1.\textwidth]{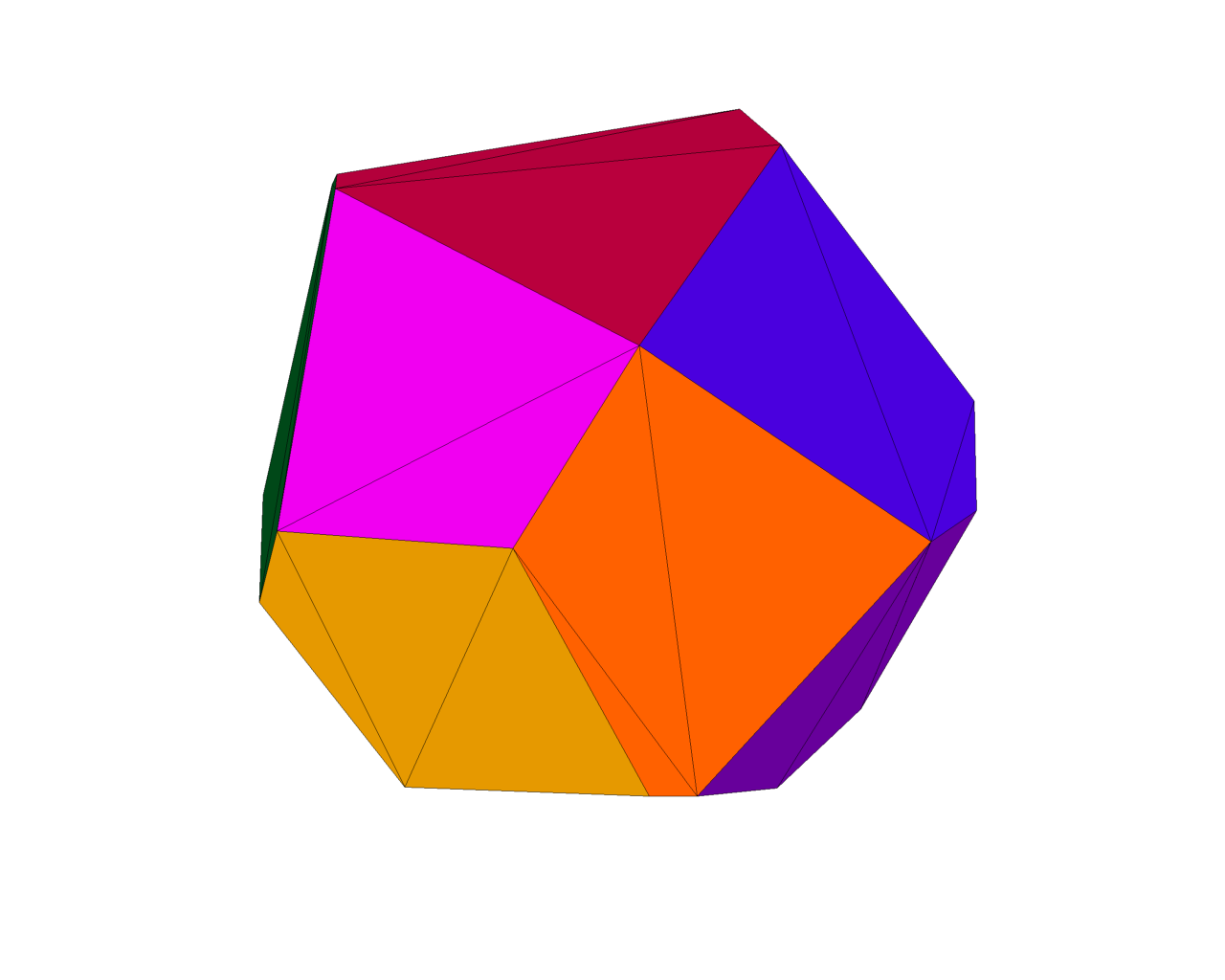}
		\caption*{$ \eta = 50 $\\ $ m = 18$}
	\end{subfigure}	
	\begin{subfigure}{0.18\textwidth}
		\includegraphics[width=1.\textwidth]{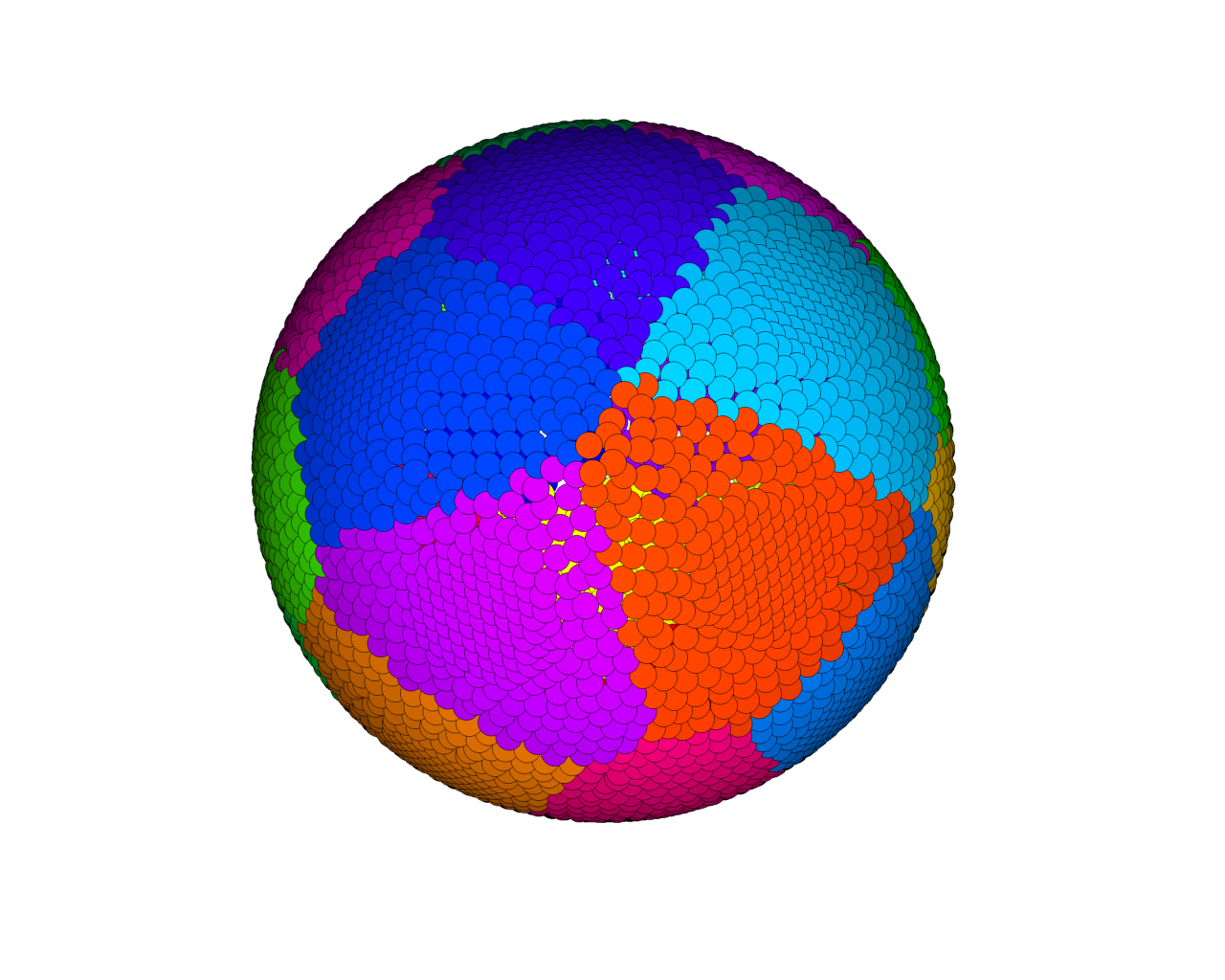}
		\includegraphics[width=1.\textwidth]{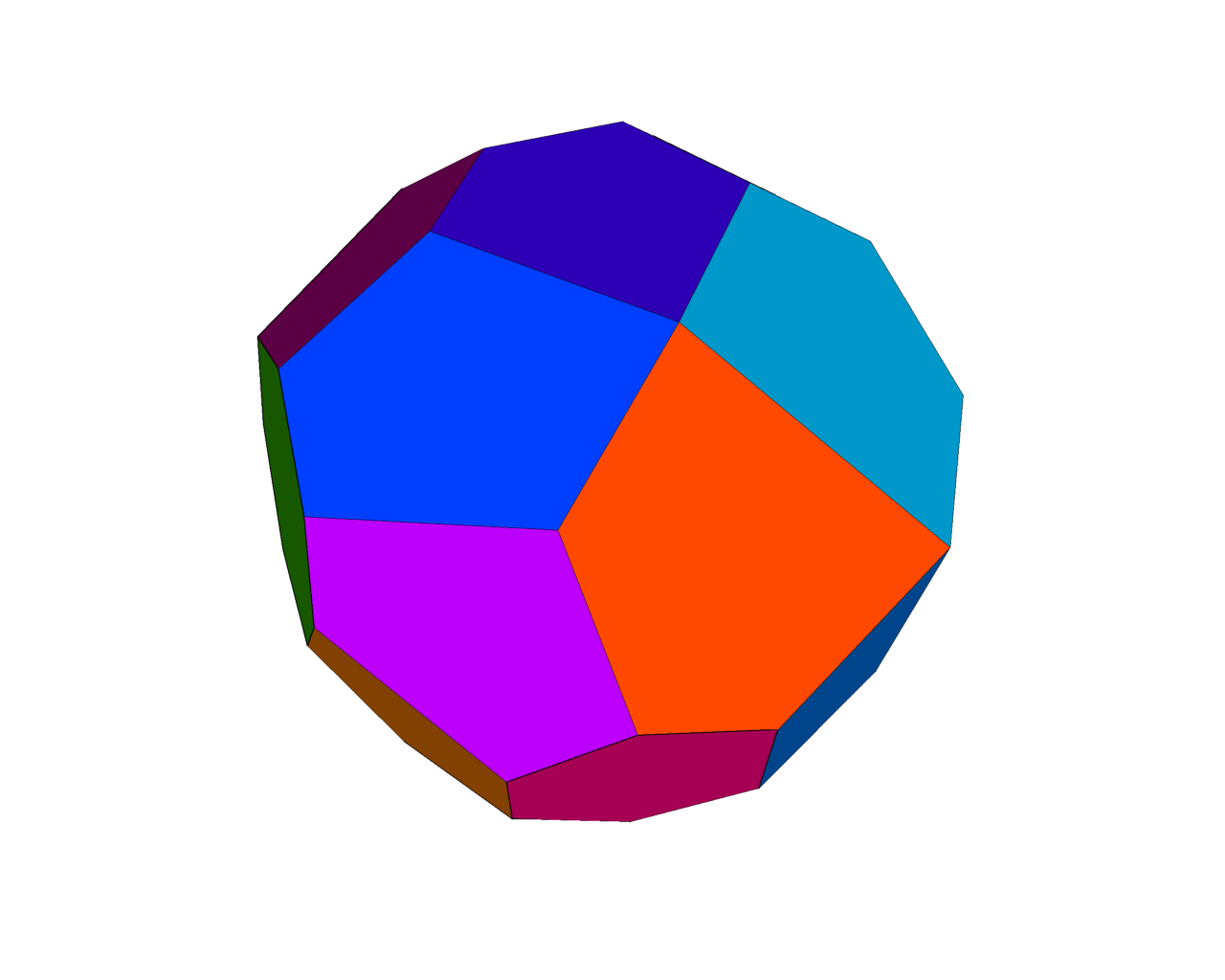}
		\includegraphics[width=1.\textwidth]{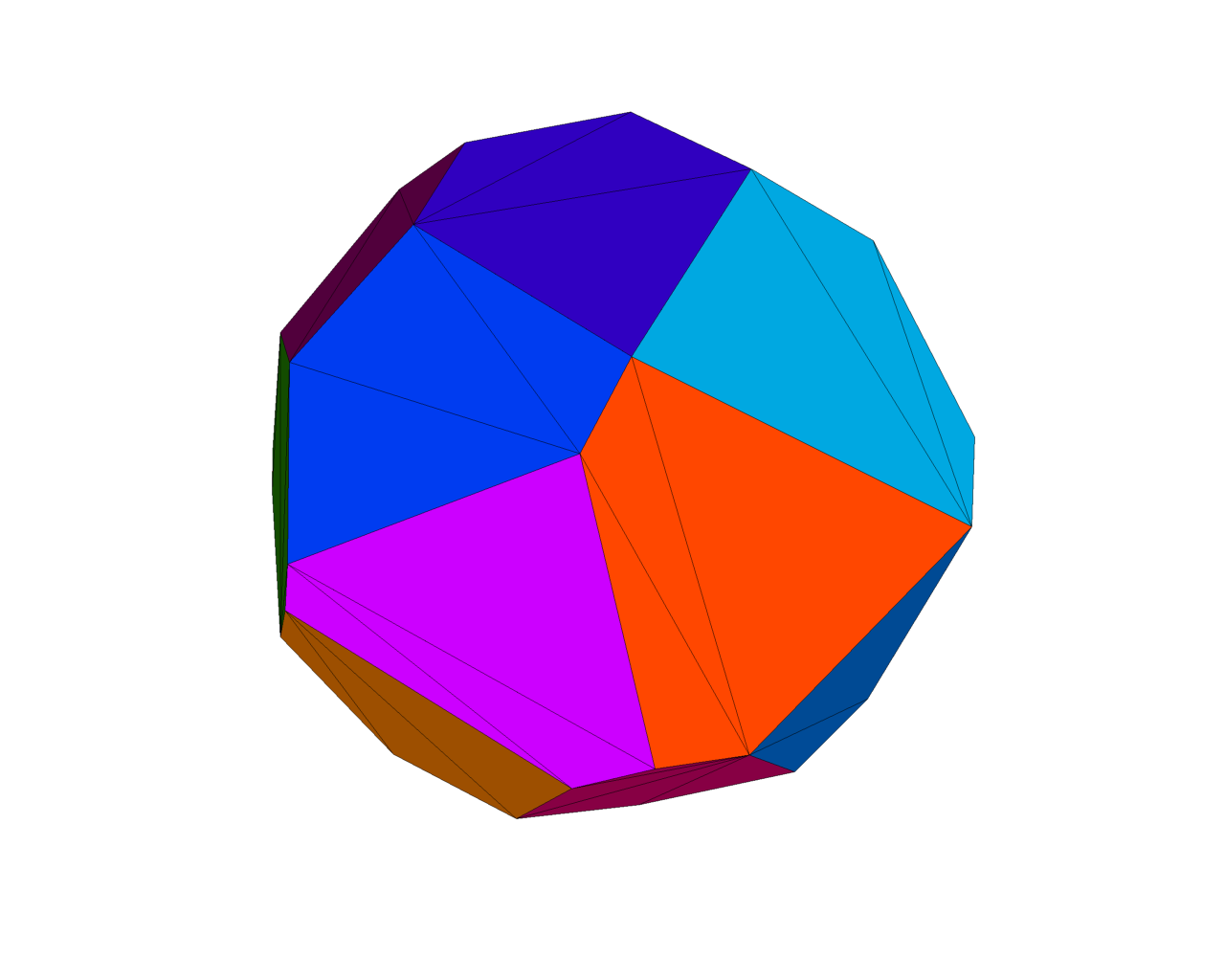}
		\caption*{$ \eta = 25 $\\ $ m = 24$}
	\end{subfigure}	
	\caption{The segmentation and simplification for different $ \eta $ values. The first row shows the segmented point sets, the second and third rows the meshes deduced via optimization and plane intersection, respectively.}
	\label{fig:sphere_kappa}
\end{figure}

In Figure~\ref{fig:sphere_kappa}, we show segmentation and simplification results w.r.t.\@ different values of~$ \eta $ taken from~$ {\lbrace 500, 200, 100, 50, 25 \rbrace} $. 
The utilized geometry is the sphere used above. The simplification points are calculated via both approaches introduced in Section~\ref{sec:Simplification}, i.e.\@ via intersection of the proxy planes and via the optimization problem given in Equation~(\ref{equ:MinimizationTangentialPlane}). 
All results are obtained from the same set of six randomly selected seed points. \edit{For the optimization case $ \eta = 500 $ we increased the weight $ \tilde{w}_i = 3 $, as otherwise the simplification points would have produced a smaller version of the cube. Hence, in this case we forced the optimization putting more emphasis on less proxy normal deviation.}
Figure~\ref{fig:sphere_kappa} also shows the amount of final proxies in correspondence to the chosen value of~$ \eta $. 
It is not surprising that with a decreasing number of~$ \eta $, the number of proxies increases as this decreases the error measure~(\ref{equ:L21energyPoints}) in order to meet the prescribed threshold. 
Note further that the resulting meshes contain vertices where more than three proxies meet, see the fourth and fifth column in Figure~\ref{fig:sphere_kappa}. 
While it is not problematic in this case, it does cause problems for a different model, see Section~\ref{sec:NoisyModel}.



\subsubsection{Face Reconstruction on the Fandisk Model}
\label{sec:FandiskModel}

\edit{
\noindent We proceed to discuss a more involved geometry, namely the Fandisk model (CAD) with~$ {\aleph = 38,840} $ vertices, shown in Figure~\ref{fig:fandisk}. 
Here, we started the segmentation with~$ 36 $ manually selected seeds, $ {\eta = 75} $, and without using splits or merges. 
We consequently obtained~$ {m = 36} $ proxy regions, shown in Figure~\ref{fig:fandiskPts}.

Reconstructing this model is challenging to our algorithm in two aspects. First, our simplification procedure requires star-convex faces, see Section~\ref{sec:SimplificationFaces}. However, the orange front plate of the fandisk model is not star-convex with respect to its barycenter (Figure~\ref{fig:starConvex}, bottom) and thus a first automatic reconstruction is slightly faulty (Figure~\ref{fig:fandiskSimpError}). 
These errors are easily identified and fixed by assigning a correct order to the contributing face vertices. 
See Section~\ref{sec:Conclusion} for a discussion how to circumvent the requirement of star-convex proxies.
The second challenging aspect is caused by the sensitivity of neighborhood notions for different densities in the point set. 
For example, the light-purple region fits between the blue and purple and hence, they see each other (Figure~\ref{fig:fandiskSimpErrorRed}, right marked spot). 
However, their planes are almost parallel, and so their intersection appears as an outlier. 
This could be avoided, if we forbid intersection points built by almost parallel planes or if we forbid those intersections that lie too far away from either one of the proxies. 
The behavior of a misplaced intersection point is also the case for the one produced by the light green, light purple and purple proxies (Figure~\ref{fig:fandiskSimpError} left marked spot), whereas in consequence a gap results between the light purple and blue area, which should not be there, according to the segmentation. 
As before, we manually removed faulty intersection points and reset face incidences to obtain a clean mesh for visual representation (Figure~\ref{fig:fandiskSimp}).

Note that these challenges are unique to the setting of point sets. 
In the context of meshed geometries, the intersection vertices can simply be ordered along the boundary of their respective proxy region, yielding a feasible face.
Also, neighborhood relations in the mesh setting can be computed via shared edges and do not require an additional neighborhood parameter~$k$.
Hence, the works of~\cite{cohen2004variational,yan2006quadraicSurfExVSA} did not have to tackle these issues, while the work of~\cite{lee2016feature} does not contain any description of how they solved these problems.
}

\begin{figure}
	\centering
	\captionsetup[subfigure]{justification=centering}
	\begin{subfigure}[t]{0.23\textwidth}
		\includegraphics[width=\textwidth]{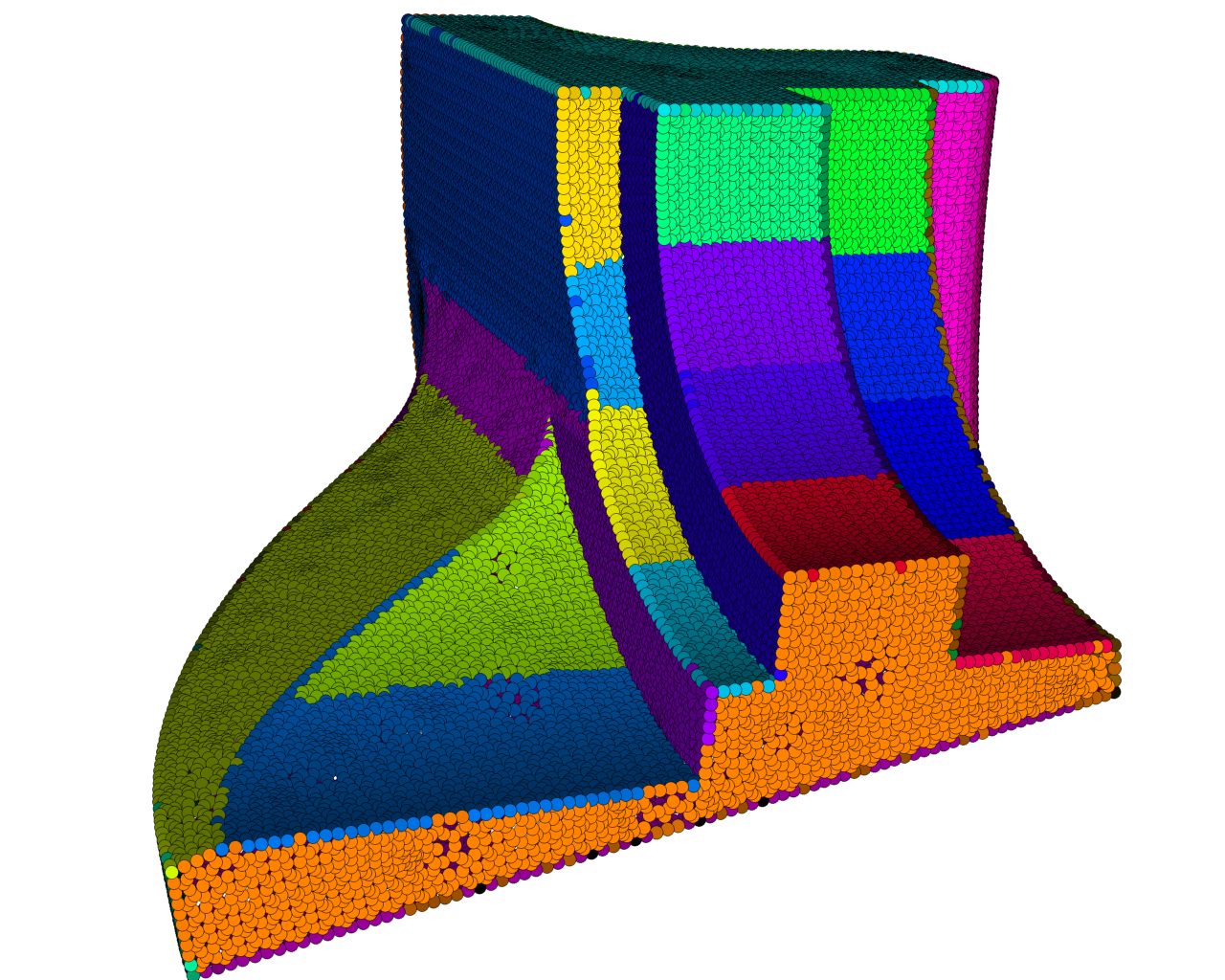}
		\caption{Initial Segmentation}
		\label{fig:fandiskPts}
	\end{subfigure}
	\begin{subfigure}[t]{0.23\textwidth}
		\includegraphics[width=\textwidth]{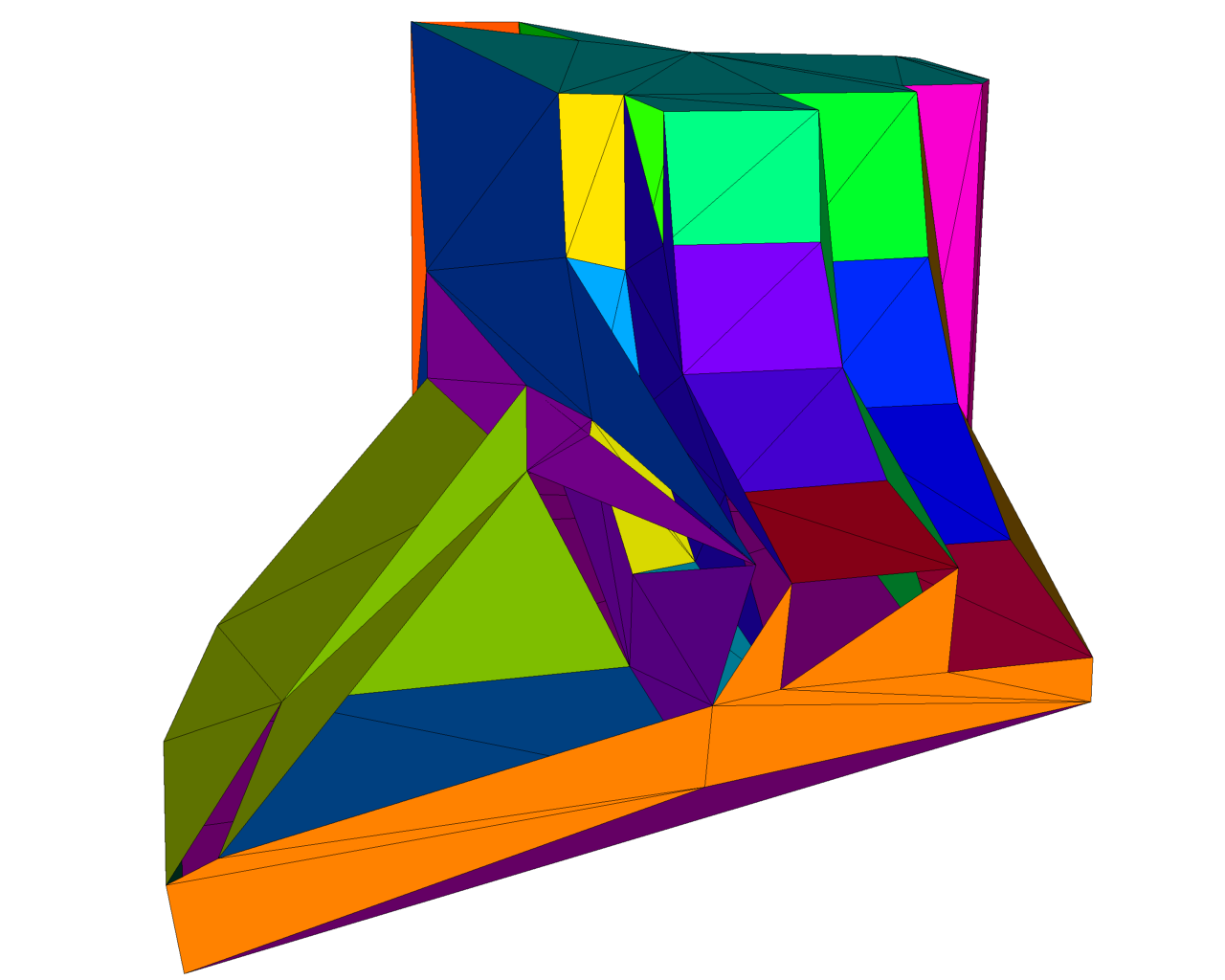}
		\caption{Faulty Vertex Order in Faces}
		\label{fig:fandiskSimpError}
	\end{subfigure}
	\begin{subfigure}[t]{0.23\textwidth}
		\includegraphics[width=\textwidth]{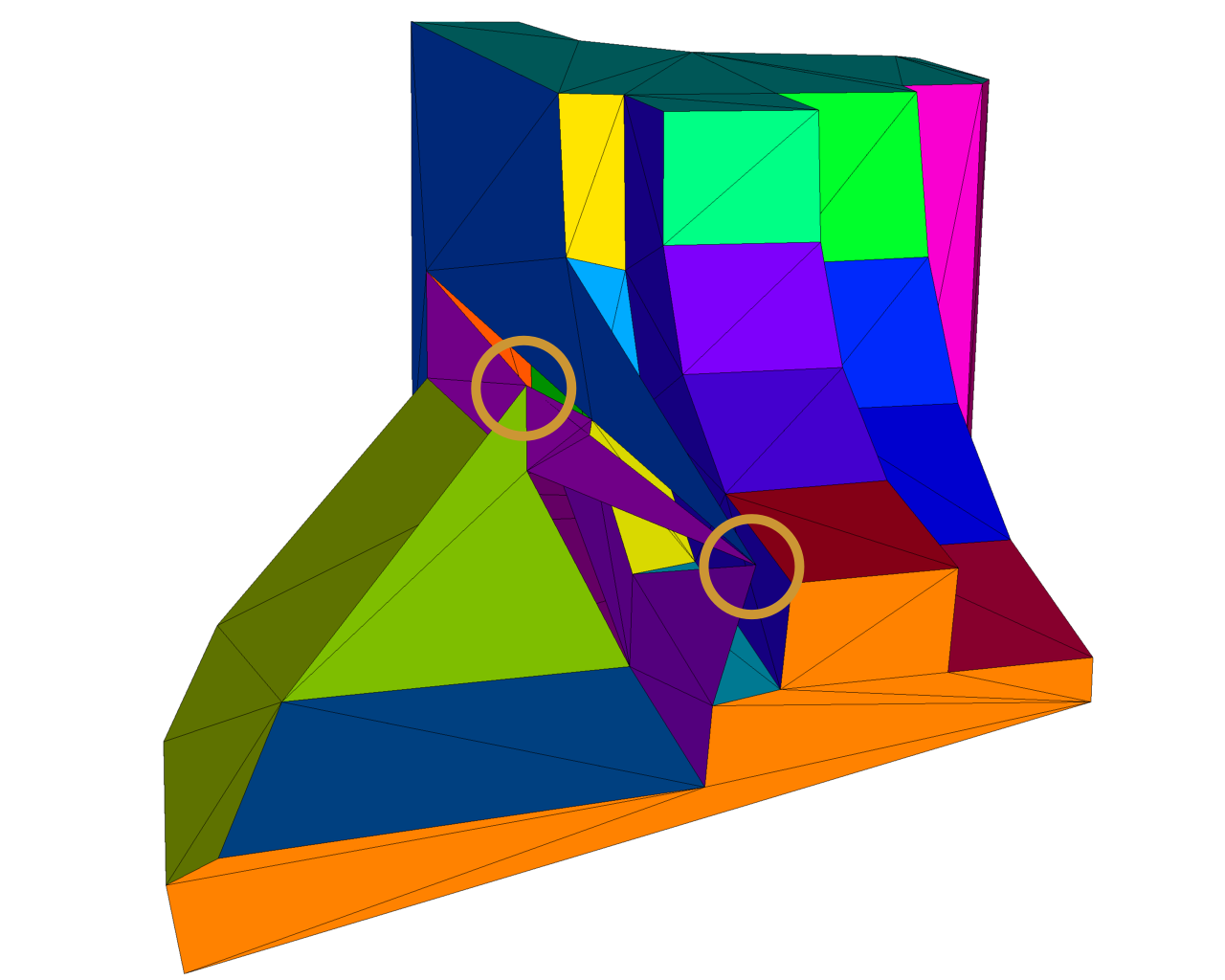}
		\caption{Faulty Proxy Intersection Points}
		\label{fig:fandiskSimpErrorRed}
	\end{subfigure}
	\begin{subfigure}[t]{0.23\textwidth}
		\includegraphics[width=\textwidth]{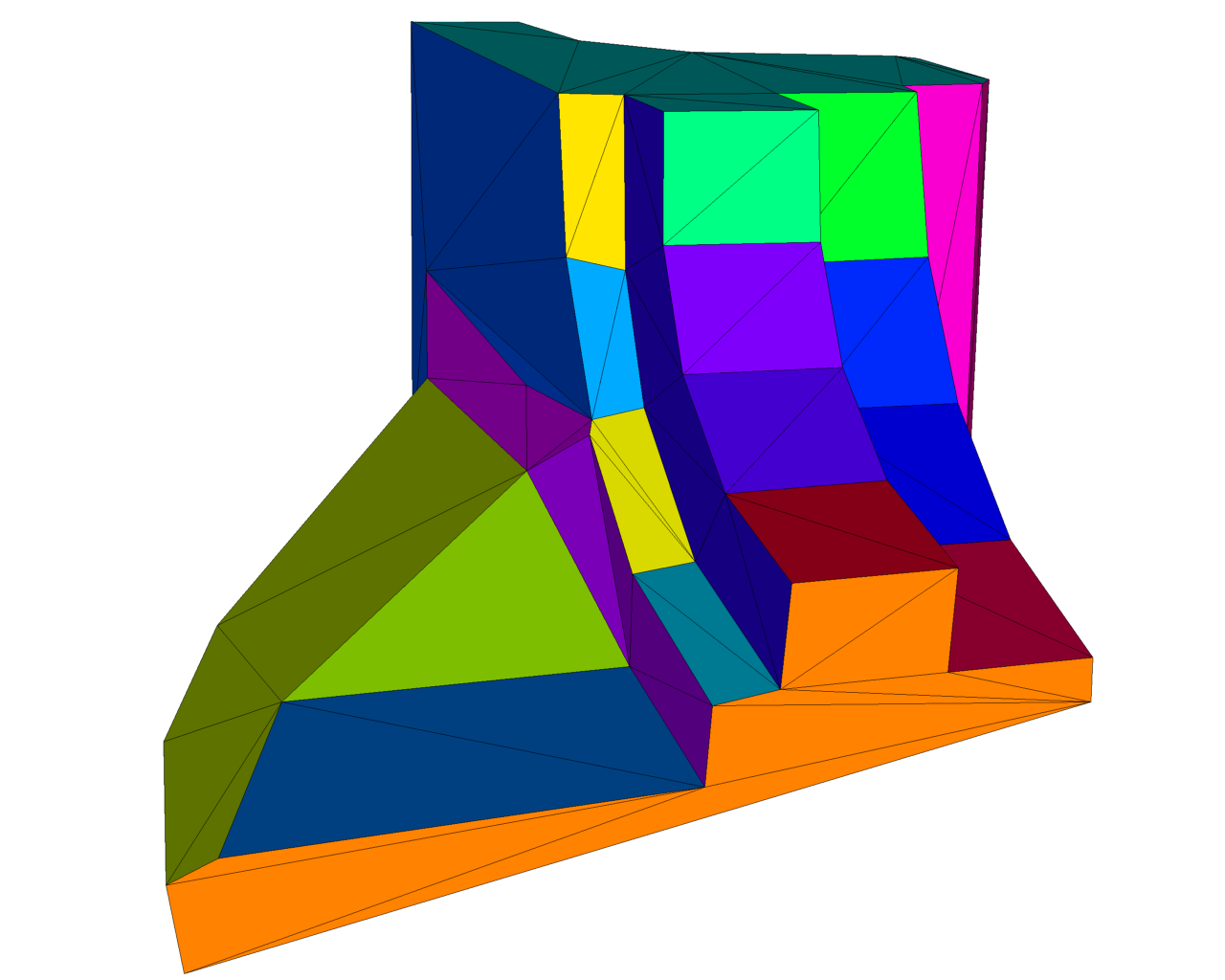}
		\caption{Corrected Simplification}
		\label{fig:fandiskSimp}
	\end{subfigure}
	\caption{Segmentation and simplification (plane intersection) of the Fandisk model.}
	\label{fig:fandisk}
\end{figure}

\subsubsection{Robustness against Noise on the Dodecahedron Model}
\label{sec:NoisyModel}

\edit{\noindent As our final experiment, we consider the simplification of a dodecahedron equipped with Gaussian noise in normal direction with an amplitude of $ 25\% $ of the average neighbor distance (taken as averaged sum over all points and their $ 12 $ nearest neighbors.).
This geometry is not easily translated into a clean mesh and therefore, the methods of~\cite{cohen2004variational,yan2006quadraicSurfExVSA} cannot be applied here straightforward.
The model consists of~$ {\aleph = 962} $ and we started with~$ 12 $ randomly chosen seeds and a threshold of~$ {\eta = 50} $. Here, in contrast to the other experiments, we use a neighborhood size of~$ {k = 12} $, because of the involved noise components. The otherwise used value of~$ {k=8} $ caused points to not be associated to any proxies. Furthermore, we allowed for splits and merges. The algorithm converged after~$ 8 $ iterations with~$ {m = 11} $ final proxies, see Figure~\ref{fig:noise_dodeca}. Observe that the faces reflecting the top proxy in the third image are not planar, which is a possible occurrence outlined in the intersection of planes when finding the simplified mesh vertices. In contrast, the optimization provides planar patches (rightmost image).

The segmentation reflects the different parts of the geometry correctly. This probably results from the normal differences caused by the noise still being smaller than the normal differences between the different faces of the dodecahedron. Hence, if the noise level and its influence in normal deviation still lies beyond the normal deviation of neighboring geometry regions, its segmentation will reflect the geometric structure well. However, this still depends on initial seed placements and therefore also an performing splits and merges.

With a segmentation reflecting the correct structure of the geometry, the simplification should not cause any additional issues, as it is the result of proxy plane intersections. 
Only the neighborhood relation between proxies might be more involved, as point locations now deviate more because of additional noise components.}

\begin{figure}
	\centering
	\captionsetup[subfigure]{justification=centering}
	\includegraphics[width=0.24\textwidth]{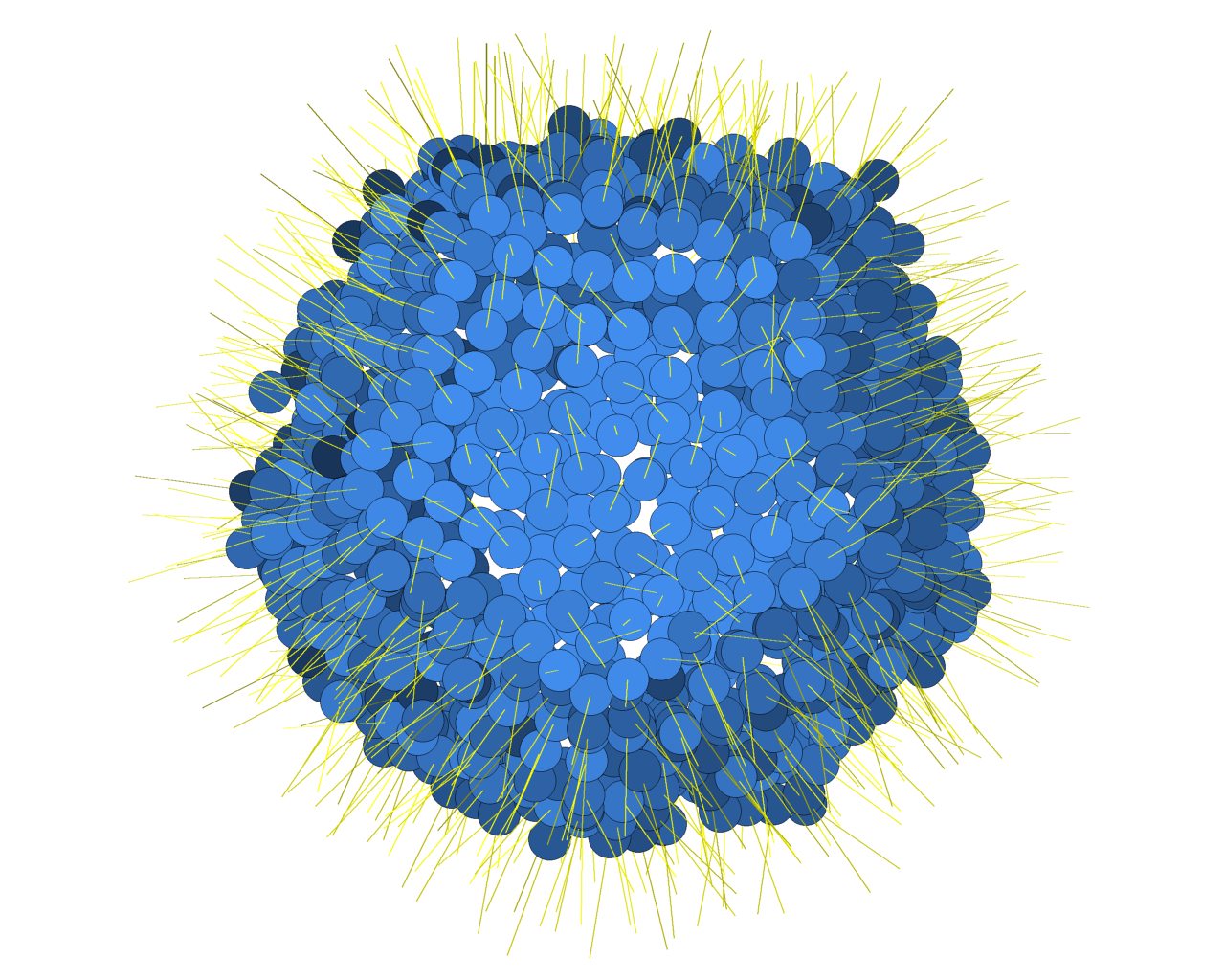}
	\includegraphics[width=0.24\textwidth]{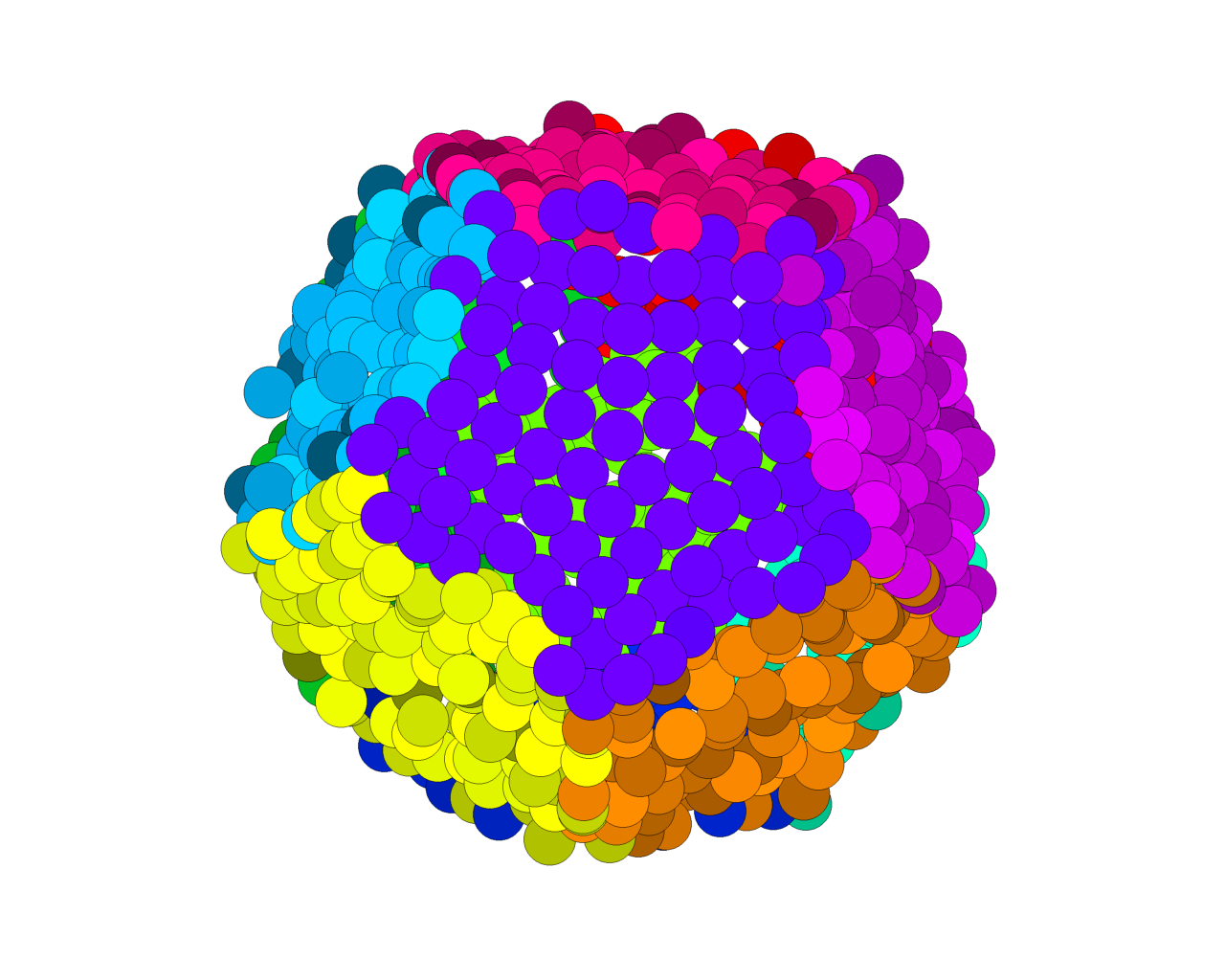}
	\includegraphics[width=0.24\textwidth]{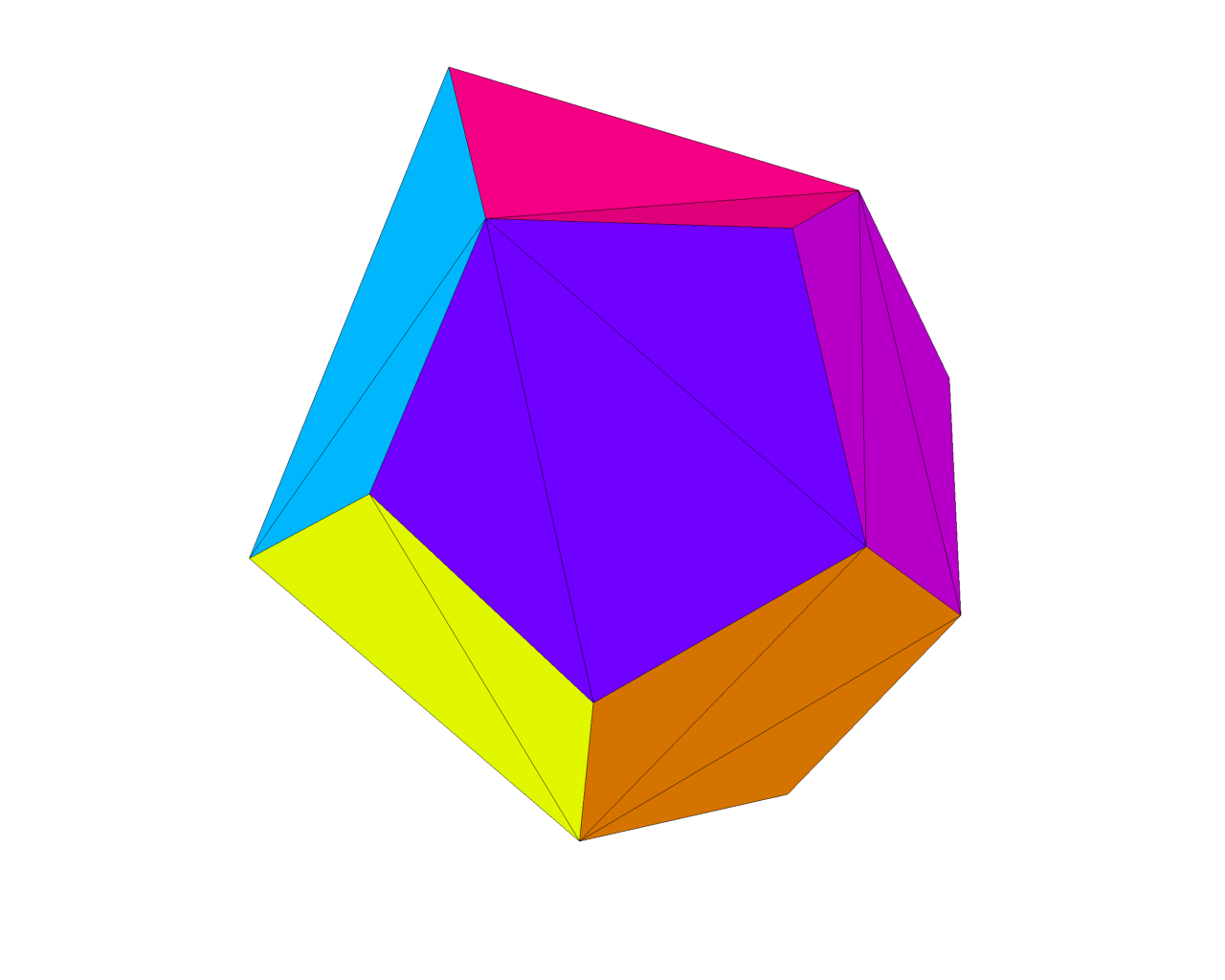}
	\includegraphics[width=0.24\textwidth]{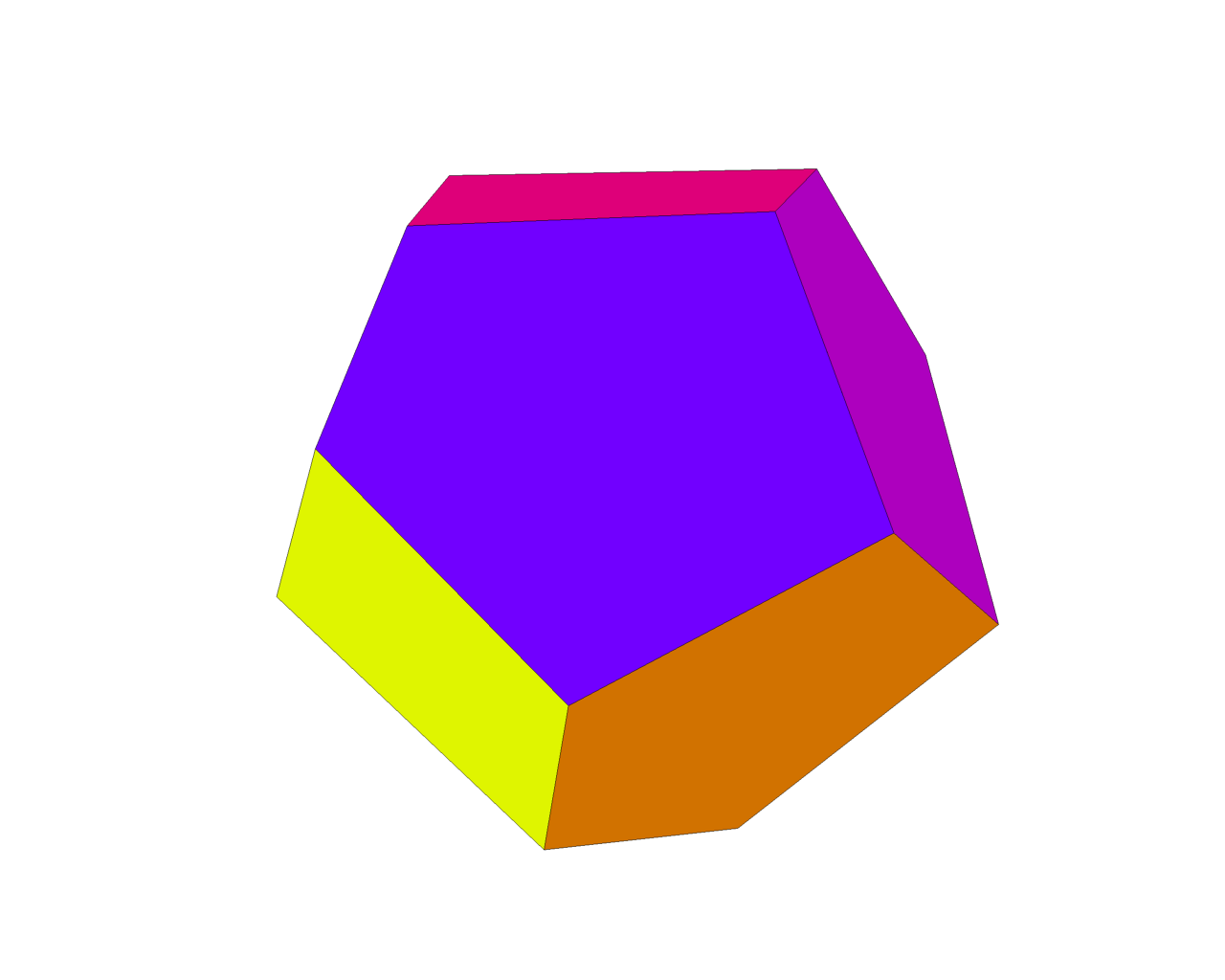}
	\caption{Noisy Dodecahedron, its segmentation and simplifications (planar intersection and optimization).}
	\label{fig:noise_dodeca}
\end{figure}


\section{Conclusion and Further Research} 
\label{sec:Conclusion}

\noindent We have shown in this paper that variational shape approximation is an effective approach to obtain a simplified mesh not only from meshed input, but also from geometries sampled by point sets. 
The presented example for non-convergence of the VSA method as used in~\cite{cohen2004variational,yan2006quadraicSurfExVSA,lee2016feature} was successfully circumvented by the introduction of a new \emph{switch} operation for which we proved convergence. 
Furthermore, by two more operations in the pipeline, namely \emph{split} and \emph{merge}, we eliminate the dependency of both the number and placement of initial seed points. 
Finally, we give a detailed description on how to obtain a simplified mesh from the segmented point set by building on the method of tangent plane intersection. Several directions are left for further research.

\edit{
First, on a theoretical level, we have shown that the introduction of the \emph{switch} operation results in guaranteed convergence. However, it remains unclear whether other alterations of the pipeline exist that came with the same result and do not affect the runtime of the algorithm as heavily as the \emph{switch} operation does.

Second, concerning the parameters, we currently do not provide any theoretical reasons for the choice of weights~$\omega_j$ in Equation~(\ref{equ:L21energyPointsOneProxy}) or weights~$\tilde{w}_j$ in Equation~(\ref{equ:TangentialPlaneEnergy}). A better understanding of these weights, aside from the experimental values used in the paper, is desirable. Similarly, the sum of normal differences~$\eta$ is not directly related to the curvature of a proxy, as it depends on the number of points contributing to the sum. Here, a threshold should be found that is independent of the number of points.

Third, our simplification process can currently not handle planar patches which are not star convex. An idea is to include border-detection algorithms for point sets to find both outer and possible inner borders---resulting from holes---of the points associated to the proxy. Given these, the mesh vertices can be easily sorted and a planar face can be obtained while not covering the holes. 

Fourth and finally, we have presented an experiment on a noisy point set. We assume that the treatment of meshes equipped with noise should be equally possible and yield even better results because of the explicit connectivity. To investigate this behavior is also left as future work.
}


\section*{Acknowledgments}

\noindent This research was supported by the DFG Collaborative Research Center TRR 109, ``Discretization in Geometry and Dynamics'', the BMS, ECMATH, and the German National Academic Foundation.

\edit{The authors would like to thank the anonymous reviewers for their comments which lead to a significant improvement of the paper.}

\appendix

\edit{
\section{Interpretation of parameter~$\eta$}
\label{app:InterpretationOfEta}

In Section~\ref{sec:UserControllerLevelOfDetail}, we introduced the user-chosen parameter~$\eta$. It relates to the energy~$\mathcal{L}^{2,1}$ as presented in Equation~(\ref{equ:L21energyPointsOneProxy}). From the definition of~$\mathcal{L}^{2,1}$, it is clear that two factors contribute to the value~${\mathcal{L}^{2,1}(P_i,N_i)}$ a given proxy~$P_i$ can achieve. These are:
\begin{itemize}
	\item The number of points~$p_j$ assigned to the proxy and
	\item the Euclidean distance of the proxy normal~$N_i$ to the respective point normals~$n_j$.
\end{itemize}
That is, a proxy can achieve a low energy by either exhibiting low deviation in its normals or by containing a low number of points. In particular the latter aspect highly depends on the number of points and the point densities in the considered model. Therefore, no general values of~$\eta$ can be presented in this paper, but the user has to choose an appropriate value for the given setup. In the following, we present a simple heuristic how to make an (initial) choice for~$\eta$.

Any point on a smooth surface can be approximated via a quadric, i.e.\@ a (hyperbolic) paraboloid~\cite{dai1998hyperbolic}. As we handle mostly (locally) convex objects, we consider an elliptic paraboloid as a simple model for a curved surface, parameterized as
\begin{align*}
	{\mathcal{P}=(u,v,\frac{u^2}{a^2}+\frac{v^2}{b^2})},
\end{align*}
then it has mean curvature
\begin{align}
\label{equ:ParaboloidCurvature}
	H(u,v) = \frac{a^2+b^2+\frac{4u^2}{a^2}+\frac{4v^2}{b^2}}{a^2b^2\sqrt{\left( 1+\frac{4u^2}{a^4}+\frac{4v^2}{b^4} \right)^3}}.
\end{align}
A normal to~$\mathcal{P}$ at~${{(u,v)}}$ is given as
\begin{align*}
	\mathcal{P}_u \times \mathcal{P}_v = \left(\begin{array}{c}1\\0\\2u/a^2\end{array}\right) \times \left(\begin{array}{c}0\\1\\2v/b^2\end{array}\right) = \left(\begin{array}{c}-2u/a^2\\-2v/b^2\\1\end{array}\right).
\end{align*}
Hence, after normalization, the point parameterized at~${(u,v)}$ contributes the following value to~$\mathcal{L}^{2,1}$, when assuming that the points are distributed uniformly on the paraboloid and therefore the proxy normal is just~${N_i=(0,0,1)^T}$:
\begin{align*}
	& \left\|\frac{1}{\sqrt{4u^2/a^4+4v^2/b^4+1}}\left(\begin{array}{c}-2u/a^2\\-2v/b^2\\1\end{array}\right) - \left(\begin{array}{c}0\\0\\1\end{array}\right)\right\|^2\\
	= & 2-\frac{2}{\sqrt{4u^2/a^4+4v^2/b^4+1}}.
\end{align*}
Placing a number~$\aleph_i$ of points regularly on the domain~${[-1,1]\times[-1,1]}$, i.e.\@ choosing ${u=j/\nu_i}$, ${v=\ell/\nu_i}$ for $ {j,\ell=1,\ldots,\nu_i} $ and $ {\nu_i:=\lceil\frac{\sqrt{\aleph_i}-1}{2}\rceil} $, we can compute the total value of~$\mathcal{L}^{2,1}$ for these points, depending on the curvature prescribed by~${(a,b)}$ as
\begin{align}
\label{equ:etaHeuristic}
	\mathcal{L}^{2,1}(P_i,N_i) = 2(2\nu_i+1)^2-\sum_{j=-\nu_i}^{\nu_i}\sum_{\ell=-\nu_i}^{\nu_i}\frac{1}{\sqrt{\left(\frac{j}{\nu_ia^2}\right)^2 +\left(\frac{\ell}{\nu_ib^2}\right)^2 + \frac{1}{4}}}.
\end{align}

Now Equation~(\ref{equ:etaHeuristic}) provides a heuristic to compute an (initial) value of~$\eta$: A user of the algorithm first chooses a desired curvature, to be covered by the proxies. From this choice and a distribution on the two main curvature directions, via Equation~(\ref{equ:ParaboloidCurvature}), the parameters~${(a,b)}$ can be computed. As the user also knows the models to which the algorithm will be applied and therefore the resolution, i.e. the number of points to be included, a second choice is the number of points~$\aleph_i$ that is roughly to be covered by a single proxy. From these two choices, using Equation~(\ref{equ:etaHeuristic}), a first estimate for~$\eta$ can be computed. If the output of the algorithm is not satisfactory, the user is of course free to tune the parameter towards the desired result.
}
\bibliographystyle{elsarticle-harv} 
\bibliography{literature}


%
%
%
\end{document}